\documentclass[11pt,letterpaper]{article}
\usepackage[T1]{fontenc}
\usepackage{amsmath,amsfonts,amssymb,amsthm}
\usepackage{mathtools}
\usepackage{bbm}
\usepackage{bm}
\usepackage[margin=1.00in]{geometry}
\usepackage{natbib}
\usepackage{hyperref}
\hypersetup{breaklinks=true,colorlinks=true,linkcolor=blue,citecolor=blue,urlcolor=blue}
\usepackage{graphicx}
\usepackage{authblk}
\usepackage{tikz}
\usepackage{pgfplots}
\pgfplotsset{compat=1.18}
\usepgfplotslibrary{groupplots}

\newtheorem{theorem}{Theorem}
\newtheorem{proposition}{Proposition}
\newtheorem{corollary}{Corollary}

\theoremstyle{definition}
\newtheorem{definition}{Definition}
\newtheorem{assumption}{Assumption}
\newtheorem{example}{Example}
\newtheorem{remark}{Remark}

\newcommand{\p}{\mathbb{P}}
\newcommand{\E}{\mathbb{E}}
\newcommand{\R}{\mathbb{R}}
\newcommand{\id}{\mathbbm{1}}

\newcommand{\cx}{\le_{\mathrm{cx}}}
\newcommand{\ES}{\mathrm{ES}}
\newcommand{\VaR}{\mathrm{VaR}}
\newcommand{\Cset}{\mathbb{C}}
\newcommand{\dsquare}{\mathop{\square}\displaylimits}
\newcommand{\dboxplus}{\mathop{\boxplus}\displaylimits}
\newcommand{\dsquareC}{\mathop{\kern.4em\square\mathrlap{_{\Cset}}\kern.4em}\displaylimits}
\newcommand{\dboxplusC}{\mathop{\kern.4em\boxplus\mathrlap{_{\Cset}}\kern.4em}\displaylimits}
\newcommand{\bX}{\boldsymbol{X}}
\newcommand{\bY}{\boldsymbol{Y}}

\renewcommand{\d}{\mathrm{d}}

\DeclareMathOperator*{\esssup}{ess\,sup}

\title{Comonotonic improvement under feasibility constraints}

\author[1]{Christopher Blier-Wong}
\author[2]{Jean-Gabriel Lauzier}

\affil[1]{Department of Statistical Sciences, University of Toronto, Canada}
\affil[2]{Department of Economics, Memorial University of Newfoundland, Canada}

\date{\today}

\begin{document}
\maketitle

\begin{abstract}
Regulatory and contractual constraints on individual
exposures are standard in insurance and reinsurance
markets, but a poorly designed constraint can distort
the economic incentives of risk-averse agents.
In the unconstrained problem, the classical
comonotonic improvement theorem guarantees
Pareto-optimal allocations that are nondecreasing in
the aggregate loss.
A constraint that is not stable under risk reduction
can destroy this property. We show by example that
Value-at-Risk caps lead to optimal allocations that
are non-comonotonic in the aggregate loss.
We identify componentwise convex-order solidity
as a sufficient condition on the feasible set that
restores the comonotonic improvement under
constraints. If replacing any agent's allocation by
a less risky one preserves feasibility, then every
feasible allocation admits a feasible comonotonic
improvement for all convex-order-consistent
preferences.
This criterion covers many constraints typical in
risk management, but excludes Value-at-Risk caps and
idiosyncratic deductibles. We illustrate the implications
of our main result in a mean-variance risk-sharing
application.
\end{abstract}

\noindent \textbf{Keywords:} Risk-sharing, regulatory constraints, convex order, Pareto-optimal allocations, mean-variance risk-sharing

%% =========================================================
\section{Introduction}\label{sec:intro}

Under convex-order-consistent preferences, the classical comonotonic
improvement theorem states that every Pareto-optimal allocation of
an aggregate loss $S$ among $n$ agents admits a comonotonic
representative with the same marginal distributions. The allocation problem
therefore collapses to a one-dimensional rearrangement of $S$. The
result was established for expected-utility agents by
\citet{borch1962equilibrium}, extended preference-free by
\citet{landsberger1994comonotonic}, and sharpened to integrable
aggregate losses on atomless spaces by
\citet{ludkovski2008comonotonicity}. In parallel, the
inf-convolution tradition of \citet{barrieu2005inf},
\citet{jouini2008optimal}, \citet{filipovic2008optimal},
\citet{acciaio2007optimal}, and \citet{ruschendorf2013mathematical}
recasts the problem on spaces of monetary risk measures, and
identifies convex-order consistency as the structural feature that
forces optimal allocations to be comonotonic. 
The preference class it covers is broad. It
contains every law-invariant coherent risk measure
\citep{bauerle2006stochastic}, every law-invariant convex risk
measure on an atomless space \citep{jouini2008optimal}, and every
distortion risk measure with concave distortion
\citep{wang1997axiomatic}, a family that includes Expected
Shortfall, spectral risk measures, proportional hazard transforms,
and Wang transforms with a concave generator. 
A parallel to the comonotonic improvement theorem is well-known 
in general equilibrium theory under a different label: 
in equilibrium, each risk-averse agent's consumption bundle is 
nondecreasing in the aggregate endowment, so that Pareto-optimal 
allocations are comonotonic \citep{mas1995microeconomic}.
Comonotonic reduction
is the starting point of most optimal risk-sharing analyses.

In practical situations, risk-sharing problems are rarely unconstrained.
Capital regulation imposes ceilings on each entity's exposure based
on its own capital, assets, and reserves. Solvency II
\citep{ec2009solvency} and the Swiss Solvency Test cap one-year
changes in own funds through a high-level $\VaR$ and $\ES$
respectively. The revised Basel III market-risk framework
\citep{bcbs2019market} replaces the earlier $\VaR$ ceiling with an
$\ES_{0.975}$ at the trading-desk level. 
Reinsurance treaties impose capacity limits, quota-share
ratios, and excess-of-loss layer caps
\citep{albrecher2017reinsurance}. Catastrophe bonds and
insurance-linked securities introduce attachment points, exhaustion
points, and step-up capacities indexed to the aggregate loss
\citep{cummins2008cat}. Peer-to-peer pools can cede an upper tail to
a reinsurer \citep{denuit2020investing}, and every sharing agreement
must satisfy individual rationality, so that each agent prefers the
pool to walking away \citep{jouini2008optimal}. Each such
arrangement carves out a strict subset of the clearing allocations
and raises the question of whether comonotonic reduction still
applies.

We seek a condition on the feasible set that preserves
comonotonic reduction under constraints, and identify it as
componentwise convex-order solidity, the property that
every clearing allocation whose components are convex-order
dominated by those of a feasible allocation is itself
feasible. Informally, making any one agent's
share less risky cannot push the allocation out of the feasible
set. Once the feasible set is solid, the classical comonotonic
improvement theorem carries over with essentially the same proof;
Theorem~\ref{thm:cc:constrained}, our constrained comonotonic
improvement theorem, reduces to the unconstrained one in a single
step. Isolating this property is the substantive step, and it forces the
constrained problem back into the geometry where the classical
argument applies.

The definition is substantive because it excludes constraints
commonly written into regulation and contracts, and this has direct
consequences for the structure of the pool and for the incentives
facing its participants. Solidity fails in two qualitatively
distinct ways. The first arises when feasibility depends on
information beyond the aggregate loss. If each agent retains a
deductible on its own portfolio, feasibility is driven by
each agent's own portfolio loss, whereas the components of a comonotonic
allocation are $\sigma(S)$-measurable. 
Optimal allocations then cannot be functions of $S$ alone, and
must condition on each agent's individual loss. 
The second mode arises even when the constraint is
$\sigma(S)$-measurable. A $\VaR$ ceiling at a fixed level is the
canonical example. $\VaR$ is not convex-order consistent, so a firm
under such a cap can lower its requirement by splitting its
portfolio across subentities and shifting mass into the tail that
$\VaR$ does not see. \citet{weber2018solvency} shows that under
Solvency II a firm can shrink its $\VaR$ capital requirement
almost to zero by splitting into enough subentities. 
This mechanism echoes observations in
\citet{xia2023infconvolution} and resonates with the
counter-monotonic quantile-based optima of
\citet{embrechts2018quantilebased} and
\citet{lauzier2023pairwise, lauzier2026risk}.
\citet{bernard2009optimal} show that an insurer subject to a $\VaR$
regulatory ceiling optimally purchases reinsurance for moderate
losses but leaves large losses entirely uninsured, since coverage
above the quantile does not affect the binding constraint. This
tail-insensitivity of $\VaR$ runs directly counter to the purpose
of the regulation, and reflects the same mechanism that breaks
convex-order solidity in our setting. The regulatory
incentive therefore, points in the opposite direction from a cap based on a
convex-order-consistent risk measure: instead of rewarding each
agent for flattening the tail of its share, a $\VaR$ cap rewards
pushing mass further out. Expected
Shortfall at the same level dominates $\VaR$ and is convex-order
consistent, so replacing the cap with an $\ES$ cap at the same level
restores solidity. To the best of our knowledge,
neither the loss of $\sigma(S)$-measurability induced by
own-endowment deductibles, nor the breakdown of comonotonic
reduction under a $\VaR$ ceiling with risk-averse decision criteria,
has been isolated in the prior risk-sharing literature.

On the positive side, solidity admits a broad class of operational
constraints: deterministic caps and floors, premium budgets,
coherent risk-measure ceilings such as Expected Shortfall, and
individual rationality. Solidity is preserved under intersection,
so combinations of these are again solid, and the constrained
problem admits the same comonotonic reduction as the unconstrained
one. A complementary thread studies constrained bilateral sharing without
a common probability measure. \citet{boonen2020bilateral} consider a
decision maker with subjective expected utility facing a
rank-dependent counterparty in a setting with no aggregate loss to
share, impose deterministic exposure caps, and obtain optima that are
monotone in the likelihood ratio rather than comonotonic. Their
failure of comonotonicity is preference-side, driven by heterogeneous
beliefs and a nonconvex distortion, whereas our non-comonotonic
examples come from the constraint side under shared
convex-order-consistent preferences. We fix convex-order-consistent
preferences and identify the structural property on the feasible set
that preserves the comonotonic reduction.

As a concrete illustration, we revisit the mean-variance
risk-sharing problem. In the unconstrained case, the optimum is
proportional, with each agent's share of $S$ tied to that agent's
risk tolerance \citep{acciaio2007optimal}. Under deterministic
capacity caps, which are solid, the optimum remains proportional,
but only in pieces. Once an agent saturates its capacity, the
remaining agents absorb the residual loss in proportion to their
own risk tolerances. The sharing rule transitions to a new linear
regime at each saturation point, and the truncated-affine form shows
how a solid constraint reshapes the pool without destroying the
comonotonic structure. Replacing the cap with a $\VaR$ ceiling at
the same level breaks solidity, and we exhibit by example an
optimal allocation whose components fail to be comonotonic with
$S$.

The remainder of this paper is structured as follows.
Section~\ref{sec:setup} fixes the setup, introduces the notation,
and recalls the unconstrained comonotonic improvement theorem.
Section~\ref{sec:main} defines componentwise convex-order solidity
and proves the constrained comonotonic improvement theorem.
Section~\ref{sec:counter} exhibits three classes of constraints
that violate solidity, each yielding an optimal allocation whose
components may not be comonotonic with $S$.
Section~\ref{sec:sufficient} verifies solidity for the constraint
families typically used in practice and discusses a funding-cost
adjustment on the objective side. Section~\ref{sec:mv} applies the 
constrained reduction to the mean-variance problem.

%% =========================================================
\section{Setup and preliminaries}\label{sec:setup}

Let $(\Omega,\mathcal{F},\p)$ be an atomless probability space. Each
agent $i\in[n]:=\{1,\dots,n\}$ is endowed with a standalone loss
$\zeta_i\in L^1$, and the aggregate loss shared by the pool is
$S:=\sum_{i=1}^n\zeta_i\in L^1$. Each agent evaluates losses through a
law-invariant risk measure $\rho_i:L^1\to\R\cup\{+\infty\}$ consistent with
risk aversion, in the loss convention where higher values of $\rho_i(X_i)$
correspond to worse loss distributions. The unconstrained risk-sharing
problem is
\begin{equation}\label{eq:intro:unconstrained_problem}
\inf_{\bX\in\mathbb{A}_n(S)}\sum_{i=1}^n \rho_i(X_i); 
\qquad
\mathbb{A}_n(S):=\left\{(X_1,\dots,X_n)\in (L^1)^n:\ \sum_{i=1}^n X_i=S\ \text{a.s.}\right\}.
\end{equation}
The feasible set $\mathbb{A}_n(S)$ collects every clearing allocation of
$S$. Under mild regularity conditions, a minimizer of
\eqref{eq:intro:unconstrained_problem} is Pareto optimal up to
deterministic cash transfers between the agents.

An allocation $\bX\in\mathbb{A}_n(S)$ is comonotonic with the
aggregate loss~$S$ if there exist nondecreasing functions
$f_1,\dots,f_n:\R\to\R$ with $X_i=f_i(S)$ a.s.\ for every $i$ and
$\sum_{i=1}^n f_i(s)=s$ for $F_S$-a.e.\ $s$ (equivalently,
$\sum_{i=1}^n f_i(S)=S$ a.s.); see \citet[Proposition 4.5]{denneberg1994nonadditive}.
We denote
the set of comonotonic clearing allocations by
\[
\mathbb{A}_n^{+}(S):=\left\{\bX\in\mathbb{A}_n(S):\ X_i=f_i(S)\ \text{a.s.\ for nondecreasing } f_1,\dots,f_n:\R\to\R\right\}.
\]
For any constraint system $\Cset$, we define
$\mathbb{A}_n^{\Cset}(S):=\{\bX\in\mathbb{A}_n(S):\bX$ satisfies every
constraint in $\Cset\}$. Let
\[
\dsquareC_{i=1}^n\rho_i(S)
\ :=\ \inf_{\bX\in\mathbb{A}_n^{\Cset}(S)}\sum_{i=1}^n\rho_i(X_i)
\quad\text{and}\quad
\dboxplusC_{i=1}^n\rho_i(S)
\ :=\ \inf_{\bX\in\mathbb{A}_n^{\Cset}(S)\cap\mathbb{A}_n^{+}(S)}\sum_{i=1}^n\rho_i(X_i)
\]
denote the constrained inf-convolution and its comonotonic
restriction, respectively. The unconstrained symbols
$\dsquare_{i=1}^n\rho_i(S)$ and $\dboxplus_{i=1}^n\rho_i(S)$ are
obtained by removing the constraint system $\Cset$. The unconstrained
reduction $\dsquare_{i=1}^n\rho_i(S)=\dboxplus_{i=1}^n\rho_i(S)$,
recalled in Theorem~\ref{thm:unconstrained}, is what we seek to extend
to the constrained setting. The central equality we study is whether
\begin{equation}\label{eq:intro:central_equality}
\dsquareC_{i=1}^n\rho_i(S)
\;=\;
\dboxplusC_{i=1}^n\rho_i(S),
\end{equation}
where the right-hand side is the same infimum restricted to
$\mathbb{A}_n^{\Cset}(S)\cap\mathbb{A}_n^{+}(S)$.

\begin{definition}[Pareto optimality]\label{def:pareto}
Given a feasible set $\mathcal{X}\subseteq\mathbb{A}_n(S)$, an
allocation $\bX^{\star}\in\mathcal{X}$ is \emph{Pareto optimal in
$\mathcal{X}$} if no $\bY\in\mathcal{X}$ satisfies
$\rho_i(Y_i)\le\rho_i(X^{\star}_i)$ for every $i\in[n]$ with at least
one strict inequality.
\end{definition}

Any minimizer of $\sum_{i=1}^n\rho_i(X_i)$ over $\mathcal{X}$ is Pareto
optimal in $\mathcal{X}$.

\subsection{Convex order}

The convex order formalizes the notion that one random variable is less
dispersed than another while preserving the mean.

\begin{definition}[Convex order]\label{def:cx}
For $X,Y\in L^1$, we write $Y\cx X$ if $\E[Y]=\E[X]$ and
$\E[\phi(Y)]\le\E[\phi(X)]$ for every convex $\phi:\R\to\R$.
\end{definition}

Equivalently, $Y\cx X$ if and only if $\E[Y]=\E[X]$ and
$\int_p^1 F^{-1}_Y(u)\,\d u\le\int_p^1 F^{-1}_X(u)\,\d u$ for every
$p\in[0,1]$. In the loss convention, $Y\cx X$ says $Y$ is a
mean-preserving contraction of $X$ in the sense of
\citet{rothschild1970increasing}; equivalently, in the wealth
convention, $-Y$ second-order stochastically dominates $-X$ with equal
means. Every risk-averse agent with finite mean prefers $Y$ to $X$.

\begin{assumption}[Convex-order-consistent preferences]\label{ass:cx_monotone}
Each agent's risk measure $\rho_i:L^1\to\R\cup\{+\infty\}$ is
convex-order consistent: $Y\cx X$ implies $\rho_i(Y)\le\rho_i(X)$.
\end{assumption}

Throughout we use the extended-arithmetic conventions
$a+(+\infty)=+\infty$ for $a\in\R\cup\{+\infty\}$, and we write
$\sum_{i=1}^n\rho_i(X_i)=+\infty$ as soon as $\rho_i(X_i)=+\infty$ for
some~$i$. Infima over empty sets are $+\infty$.

Assumption~\ref{ass:cx_monotone} covers expected
concave-utility functionals in the loss convention, i.e.\
$X\mapsto\E[\ell(X)]$ with $\ell$ convex and nondecreasing (equivalently
$X\mapsto-\E[u(-X)]$ for concave $u$, where $-X$ is the wealth
representation); it also covers law-invariant coherent risk measures
\citep{bauerle2006stochastic}, distortion risk measures with concave
distortion \citep{wang1997axiomatic}, and law-invariant convex monetary
risk measures on atomless probability spaces \citep{jouini2008optimal},
under the standard regularity conditions on the underlying function
space that render these classes convex-order consistent.

Throughout the paper, $\VaR_\alpha$ denotes the lower
$\alpha$-quantile,
$\VaR_\alpha(X):=\inf\{x\in\R:\p(X\le x)\ge\alpha\}$, for
$\alpha\in(0,1)$.

\subsection{Unconstrained comonotonic improvement}

The following is the comonotonic improvement theorem of
\citet{ludkovski2008comonotonicity}, which extends to integrable
aggregate losses the construction of
\citet{landsberger1994comonotonic} for the bounded discrete case.
The original statements are phrased in terms of a stop-loss order;
for $L^1$ aggregate losses with equal means, this coincides with
the convex order used here.

\begin{theorem}\label{thm:unconstrained}
Let $(\Omega,\mathcal{F},\p)$ be atomless and $S\in L^1$. For every
$\bX=(X_1,\dots,X_n)\in\mathbb{A}_n(S)$ there exists
$\bar{\bX}=(\bar X_1,\dots,\bar X_n)\in\mathbb{A}_n^{+}(S)$ such that
$\bar X_i\cx X_i$ for every $i\in[n]$.
\end{theorem}

Comonotonicity also has a desirable incentive interpretation,
sometimes called the no-sabotage property.
If $\bX\in\mathbb{A}_n^{+}(S)$,
so each share is a nondecreasing function of the aggregate loss
$X_i=f_i(S)$, then an agent cannot reduce its own allocation by
secretly inflating its own contribution to $S$, since a larger $S$
weakly increases every $X_i$. See \citet{carlier2003pareto} for
further discussion of this property and its economic relevance.
Non-comonotonic allocations generally
lack this stability property.

\begin{remark}\label{rem:two_features}
Two features of the comonotonic improvement are relevant to Section~\ref{sec:main}.
First, it is componentwise: $\bar X_i\cx X_i$ holds separately for
each $i$, with no cross-component comparison. Second, it is
law-based: $\bar X_i\cx X_i$ compares the marginal distributions of
$\bar X_i$ and $X_i$ alone and imposes no restriction on their
pathwise coupling. Both features are preserved by feasibility
constraints expressed through componentwise convex-order-consistent
functionals of each $X_i$.
\end{remark}

%% =========================================================
\section{The constrained comonotonic improvement theorem}\label{sec:main}

This section identifies the structural condition on the feasible set
under which the unconstrained comonotonic improvement theorem remains
valid. We first introduce componentwise convex-order solidity as a
property of the constrained feasible set, and then state and prove the
constrained comonotonic improvement theorem.

\subsection{Componentwise convex-order solidity}\label{sec:solidity}

The comonotonic improvement in Theorem~\ref{thm:unconstrained}
replaces each component $X_i$ by a convex-order smaller $\bar X_i$.
What a feasible set must tolerate for the comonotonic improvement
to stay inside it is therefore exactly this operation: any componentwise
convex-order reduction of a feasible allocation should remain
feasible.

\begin{definition}[Componentwise convex-order solid feasible set]\label{def:cc:solid}
A feasible set $\mathcal{X}\subseteq\mathbb{A}_n(S)$ is \emph{componentwise
convex-order solid} if, for every $\bX\in\mathcal{X}$ and every
$\bY\in\mathbb{A}_n(S)$,
\[
Y_i\cx X_i\ \text{for every } i\in[n]\quad\Longrightarrow\quad \bY\in\mathcal{X}.
\]
\end{definition}

Solidity ensures that the comonotonic rearrangement of
Theorem~\ref{thm:unconstrained}, which replaces each component by a
convex-order smaller one, cannot leave the feasible set. The
condition is componentwise in convex order, but joint through the
clearing constraint $\bY\in\mathbb{A}_n(S)$, and it is preference-free,
in that the reduction applies to every convex-order-consistent
preference family once $\mathcal{X}$ is solid. Solidity is also
preserved under arbitrary intersections, by direct verification of
Definition~\ref{def:cc:solid}.
Marginal constraints are compatible only when their admissible marginal
sets are downward closed in convex order; in particular, ceilings based
on convex-order-consistent law-invariant functionals are compatible,
whereas law-invariance alone is not enough, with Value-at-Risk ceilings
in Section~\ref{sec:counter} as the canonical counterexample. By
contrast, constraints coupling~$X_i$ to variables outside~$\sigma(S)$
are typically not compatible.
Section~\ref{sec:counter} exhibits three failure modes and
Section~\ref{sec:sufficient} catalogues the admissible families.

\subsection{The constrained comonotonic improvement theorem}

The main result of this paper shows that componentwise convex-order
solidity is sufficient for the comonotonic improvement in
Theorem~\ref{thm:unconstrained} to remain within the feasible set.

\begin{theorem}\label{thm:cc:constrained}
Let $(\Omega,\mathcal{F},\p)$ be atomless, let $S\in L^1$, and let
$\mathcal{X}\subseteq\mathbb{A}_n(S)$ be nonempty and componentwise
convex-order solid. Then:
\begin{enumerate}
\item[(i)] for every $\bX\in\mathcal{X}$ there exists a comonotonic
allocation $\bar{\bX}\in\mathcal{X}\cap\mathbb{A}_n^{+}(S)$ with
$\bar X_i\cx X_i$ for every $i\in[n]$;
\item[(ii)] if $\rho_1,\dots,\rho_n$ satisfy
Assumption~\ref{ass:cx_monotone}, then
\begin{equation}\label{eq:cc:thm:inf}
\inf_{\bX\in\mathcal{X}}\sum_{i=1}^n\rho_i(X_i)
\ =\
\inf_{\bX\in\mathcal{X}\cap\mathbb{A}_n^{+}(S)}\sum_{i=1}^n\rho_i(X_i);
\end{equation}
\item[(iii)] under Assumption~\ref{ass:cx_monotone}, if
$\bX^{\star}\in\mathcal{X}$ is Pareto optimal, then there exists
$\bar{\bX}^{\star}\in\mathcal{X}\cap\mathbb{A}_n^{+}(S)$ with
$\rho_i(\bar X^{\star}_i)=\rho_i(X^{\star}_i)$ for every $i$.
\end{enumerate}
\end{theorem}

\begin{proof}
Fix $\bX\in\mathcal{X}$. By Theorem~\ref{thm:unconstrained}, there is
$\bar{\bX}\in\mathbb{A}_n^{+}(S)$ with $\bar X_i\cx X_i$ for every~$i$.
Componentwise convex-order solidity of $\mathcal{X}$ gives
$\bar{\bX}\in\mathcal{X}$, proving~(i).

For (ii), the inequality ``$\le$'' follows from
$\mathcal{X}\cap\mathbb{A}_n^{+}(S)\subseteq\mathcal{X}$. The reverse
inequality follows from (i) and Assumption~\ref{ass:cx_monotone}:
$\sum_i\rho_i(\bar X_i)\le\sum_i\rho_i(X_i)$ for every $\bX\in\mathcal{X}$.

For (iii), apply (i) to $\bX^\star$. Assumption~\ref{ass:cx_monotone}
gives $\rho_i(\bar X^\star_i)\le\rho_i(X^\star_i)$ for every~$i$, and
Pareto optimality rules out any strict inequality.
\end{proof}

Theorem~\ref{thm:cc:constrained} provides a direct recipe: to justify
the reduction to the comonotonic subclass on a constrained problem,
one verifies componentwise convex-order solidity of the feasible set,
which decomposes over intersection; see Section~\ref{sec:solidity}. 
Section~\ref{sec:sufficient} carries this out for the constraint
families typically used in practice. When solidity fails, the comonotonic
restriction still returns an upper bound on the constrained
inf-convolution, but the bound may be strict.

%% =========================================================
\section{Constraints that preclude comonotonic reduction}\label{sec:counter}

Three examples show how comonotonic reduction can fail.
Section~\ref{ss:counter:private_deductible} shows that constraints
coupling $X_i$ with state variables outside $\sigma(S)$, such as an
agent's own deductible, measurable with respect to its individual
loss, need not be admissible.
Section~\ref{ss:counter:es_stepup} shows that
$\sigma(S)$-measurability combined with monotonicity in $S$ is not
sufficient either: a fixed retention plus a catastrophe-only capacity
top-up produces a strict gap in
\eqref{eq:intro:central_equality} under Expected Shortfall objectives.
Section~\ref{ss:counter:var} shows that a Value-at-Risk ceiling produces
a strict gap even when the initial endowments are comonotonic and
feasible.

\subsection{Private deductibles on idiosyncratic state}\label{ss:counter:private_deductible}

Consider an agent $i$ that contributes its own portfolio of
exposures to the pool. Retention rules in practice are
defined on the agent's own portfolio, so the feasibility constraint
carries information beyond the aggregate loss.
We model this by a deterministic deductible $d_i$ together with a
retention rule that references the agent's idiosyncratic endowment
$\zeta_i$. Formally, the constraint $X_i=\zeta_i$ on
$\{\zeta_i< d_i\}$ and $X_i\ge d_i$ on $\{\zeta_i\ge d_i\}$ is measurable
with respect to the idiosyncratic $\sigma$-algebra
$\mathcal{G}_i:=\sigma(\zeta_i)$. When
$\sigma(S)\subsetneq\sigma(S,\zeta_i)$, that is, when $\zeta_i$ is not
a function of $S$ alone, the constraint is pathwise and lies outside
the reach of the comonotonic improvement theorem.

\begin{example}\label{ex:counter:idiosyncratic}
Let $\zeta_1,\zeta_2$ be independent Bernoulli$(1/2)$ and set $S:=\zeta_1+\zeta_2$.
Each agent has a deductible $d_i=1$: losses below $d_i$ are fully retained,
and at least $d_i$ is retained on losses at or above $d_i$, giving
\[
\mathcal{X}:=
\left\{(X_1,X_2)\in\mathbb{A}_2(S):
X_i=0 \text{ on }\{\zeta_i=0\},\
X_i\ge 1 \text{ on }\{\zeta_i=1\},\
i=1,2
\right\}.
\]
The autarky allocation $(\zeta_1,\zeta_2)$ lies in $\mathcal{X}$, but the feasible allocations
$(1,0)$ and $(0,1)$ on $\{S=1\}$ are not $\sigma(S)$-measurable: the event
$\{S=1\}$ does not record who bore the risk.
The conditional mean $Y_i:=\mathbb{E}[\zeta_i\mid S]=S/2$ satisfies
$Y_i\le_{\mathrm{cx}}\zeta_i$, yet on $\{\zeta_1=0,\zeta_2=1\}$ the constraint
requires $X_1=0$ while $Y_1=1/2\ne 0$, so $(Y_1,Y_2)\notin\mathcal{X}$ and
$\mathcal{X}$ is not componentwise convex-order solid.
\end{example}

\begin{figure}[ht]
\centering
\begin{minipage}{0.49\linewidth}
\centering
\begin{tikzpicture}
\begin{axis}[
  width=0.95\linewidth, height=0.7\linewidth,
  ybar, bar width=7pt,
  xtick={1,1.7,3},
  xticklabels={{$(0,1)$},{$(1,0)$},{$(1,1)$}},
  xmin=0, xmax=3.7,
  ymin=0, ymax=1.65, ytick={0,1},
  ylabel={Allocated share},
  xlabel={atom $(\zeta_1,\zeta_2)$},
  title={\small Autarky $(\zeta_1,\zeta_2)$},
  legend pos=north west, legend cell align=left,
  legend style={draw=none, fill=none, font=\scriptsize},
  x axis line style={shorten >=-14pt},
  axis lines=left, tick style={black},
  clip=false,
]
\addplot[fill=blue!25, draw=blue] coordinates {(1,0) (1.7,1) (3,1)};
\addlegendentry{Agent 1}
\addplot[fill=green!30, draw=green!60!black] coordinates {(1,1) (1.7,0) (3,1)};
\addlegendentry{Agent 2}
\end{axis}
\end{tikzpicture}
\end{minipage}\hfill
\begin{minipage}{0.49\linewidth}
\centering
\begin{tikzpicture}
\begin{axis}[
  width=0.95\linewidth, height=0.7\linewidth,
  ybar, bar width=7pt,
  xtick={1,1.7,3},
  xticklabels={{$(0,1)$},{$(1,0)$},{$(1,1)$}},
  xmin=0, xmax=3.7,
  ymin=0, ymax=1.65, ytick={0,0.5,1},
  ylabel={Allocated share},
  xlabel={atom $(\zeta_1,\zeta_2)$},
  title={\small $(S/2,S/2)$: violates retention},
  x axis line style={shorten >=-14pt},
  axis lines=left, tick style={black},
  clip=false,
]
\addplot[fill=blue!25, draw=blue] coordinates {(1,0.5) (1.7,0.5) (3,1)};
\addplot[fill=green!30, draw=green!60!black] coordinates {(1,0.5) (1.7,0.5) (3,1)};
\end{axis}
\end{tikzpicture}
\end{minipage}
\caption{Allocations on the four atoms $(\zeta_1,\zeta_2)$ in
Example~\ref{ex:counter:idiosyncratic}. Autarky is feasible on every
atom but is not $\sigma(S)$-measurable, since $X_1$ takes both values
$0$ and $1$ on the layer $\{S=1\}$. The conditional mean $(S/2,S/2)$
is $\sigma(S)$-measurable but violates the retention requirement
$X_i=0$ on $\{\zeta_i=0\}$ at the atoms $(0,1)$ and $(1,0)$.}
\label{fig:counter:private_deductible}
\end{figure}
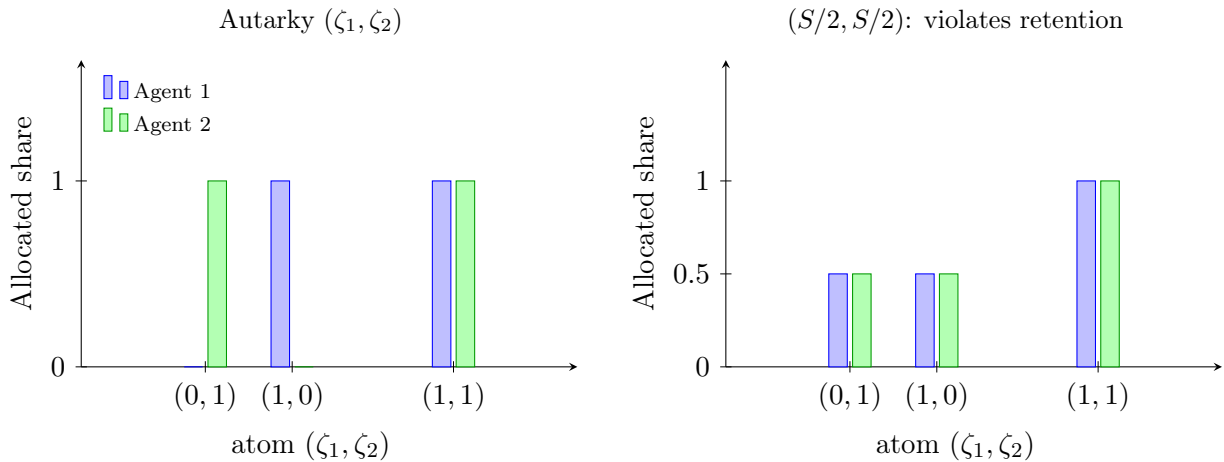

The failure in Example~\ref{ex:counter:idiosyncratic} is pathwise
rather than marginal in the sense of Remark~\ref{rem:two_features}. The
comonotonic improvement theorem produces $\bar X_1$ whose marginal
distribution is a mean-preserving contraction of that of $X_1$, but it
exercises no control over how the values of
$\bar X_1$ couple pathwise with the idiosyncratic state $\zeta_1$. The
non-solid feature of the constraint is the idiosyncratic-state retention
rule $X_1=\zeta_1$ on $\{\zeta_1<d_1\}$, not the deterministic
one-sided bound $X_1\ge d_1$ on $\{\zeta_1\ge d_1\}$: two random
variables with identical marginal distributions can differ in whether
they satisfy the equality-on-$\{\zeta_1<d_1\}$ requirement, because
the coupling with $\zeta_1$ on the idiosyncratic $\sigma$-algebra
$\mathcal{G}_1$ is not recorded by the marginal distribution of $X_1$.

The same mechanism applies to any constraint that couples $X_i$
pathwise with a state variable outside $\sigma(S)$. Examples include
covariance caps $\mathrm{Cov}(X_i,\Sigma)\le c$ with
$\Sigma\notin\sigma(S)$, conditional-expectation constraints
$\E[X_i\mid\mathcal{G}]\le c$ with
$\mathcal{G}\not\subseteq\sigma(S)$, and cross-counterparty exposure
caps $X_i\le\lambda_{ij}X_j$ a.s., which couple $X_i$ with $X_j$ rather
than with a quantity fixed by the problem data.

\subsection{Step-up capacity with fixed retention}\label{ss:counter:es_stepup}

We now exhibit a constraint system that is fully $\sigma(S)$-measurable
and monotone in $S$, yet produces a strict gap in
\eqref{eq:intro:central_equality} under Expected Shortfall preferences.
The specification combines a fixed retention
with a catastrophe-only capacity top-up arising
from contingent capital, a reinstatement line, or a government
backstop.

\begin{example}\label{ex:counter:es_stepup}
Let $S$ take the values $1,2,3$ with probability $1/3$ each, and write
$X_2=S-X_1$. The treaty fixes a deterministic retention on
non-catastrophe states and a step-up cap on the catastrophe state.
We model it as an aggregate-indexed treaty, i.e.\ require $X_1$ to be
$\sigma(S)$-measurable together with $X_1=1/4$ on $\{S\le2\}$ and
$1/4\le X_1\le 7/4$ on $\{S=3\}$. The $\sigma(S)$-measurability
restriction is part of the constraint: a treaty whose terms reference
only the aggregate loss specifies the ceded amount as a function of
$S$ alone. We return below to what happens if this restriction is
dropped. Under these constraints, every feasible allocation has the
form
$X_1=\tfrac14\,\id_{\{S\le 2\}}+a\,\id_{\{S=3\}}$ and
$X_2=\tfrac34\,\id_{\{S=1\}}+\tfrac74\,\id_{\{S=2\}}+(3-a)\,\id_{\{S=3\}}$
for some $a\in[1/4,7/4]$. Since $X_1$ is already nondecreasing in $S$,
comonotonicity reduces to $X_2$ being nondecreasing in $S$, which
requires $3-a\ge 7/4$, that is, $a\le 5/4$.

Take $\rho_1=\ES_{1/5}$ and $\rho_2=\ES_{1/3}$. A direct
computation shows that $\rho_1(X_1)+\rho_2(X_2)$ is strictly
decreasing in~$a$, so the constrained minimum is attained at
$a^\star=7/4$ (value $19/8$), while the comonotonic restriction
$a\in[1/4,5/4]$ achieves only $29/12$; the strict gap
$29/12-19/8=1/24>0$ shows \eqref{eq:intro:central_equality} fails.
The comonotonic rearrangement of the optimizer ($a=7/4$) along $S$
gives $\bar X_1=(1/4,3/4,5/4)$, which violates the fixed-retention
clause $X_1=1/4$ on $\{S\le2\}$, exhibiting the failure of componentwise
convex-order solidity.
\end{example}

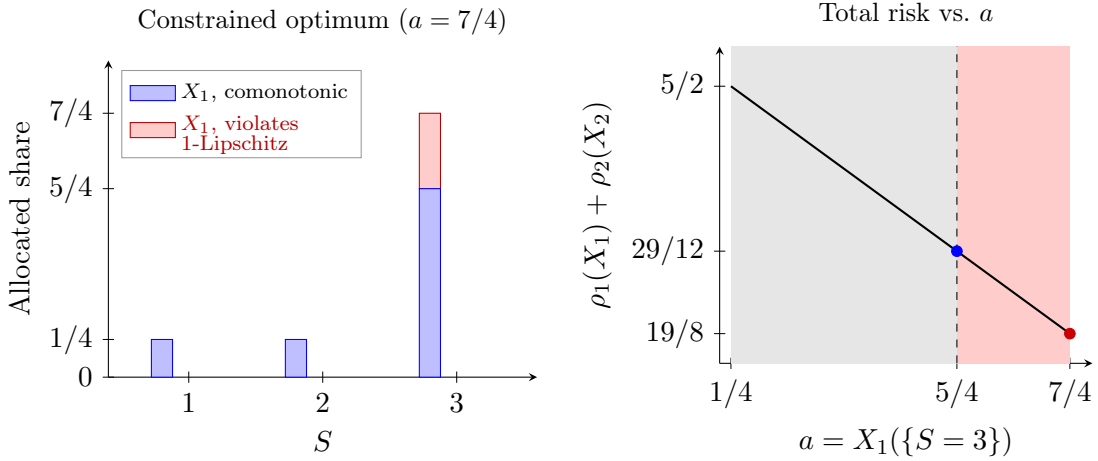
\begin{figure}[ht]
\centering
\begin{minipage}{0.45\linewidth}
\centering
\begin{tikzpicture}
\begin{axis}[
  width=0.98\linewidth, height=0.78\linewidth,
  ybar stacked, bar width=8pt,
  title={\small Constrained optimum ($a=7/4$)},
  xlabel={$S$}, ylabel={Allocated share},
  xmin=0.4, xmax=3.6,
  xtick={1,2,3},
  ymin=0, ymax=2.1,
  ytick={0,0.25,1.25,1.75},
  yticklabels={$0$,$1/4$,$5/4$,$7/4$},
  axis lines=left, tick style={black}, enlargelimits=false,
  legend pos=north west,
  legend style={font=\scriptsize, cells={anchor=west}, draw=black!40,
                fill=white, inner sep=3pt, row sep=1pt},
]
% X1 comonotonic portion
\addplot[fill=blue!25, draw=blue, forget plot]
  coordinates {(0.8,0.25) (1.8,0.25) (2.8,1.25)};
% X1 non-comonotonic excess (stacked on top at S=3)
\addplot[fill=red!20, draw=red!70!black, forget plot]
  coordinates {(0.8,0) (1.8,0) (2.8,0.5)};
\addlegendimage{fill=blue!25, draw=blue, area legend}
\addlegendentry{$X_1$, comonotonic}
\addlegendimage{fill=red!20, draw=red!70!black, area legend}
\addlegendentry{\textcolor{red!70!black}{\parbox{2.2cm}{$X_1$, violates\\[-2pt]1-Lipschitz}}}
\end{axis}
\end{tikzpicture}
\end{minipage}
\begin{minipage}{0.45\linewidth}
\centering
\begin{tikzpicture}
\begin{axis}[
  width=0.88\linewidth, height=0.78\linewidth,
  title={\small Total risk vs.\ $a$},
  xlabel={$a=X_1(\{S=3\})$},
  ylabel={$\rho_1(X_1)+\rho_2(X_2)$},
  xmin=0.2, xmax=1.85, ymin=2.36, ymax=2.52,
  xtick={0.25,1.25,1.75},
  xticklabels={$1/4$,$5/4$,$7/4$},
  ytick={2.375,2.41667,2.5},
  yticklabels={$19/8$,$29/12$,$5/2$},
  axis lines=left, tick style={black}, enlargelimits=false,
]
\fill[gray!20]  (axis cs:0.25,2.36) rectangle (axis cs:1.25,2.55);
\fill[red!20]   (axis cs:1.25,2.36) rectangle (axis cs:1.75,2.55);
\draw[dashed, black] (axis cs:1.25,2.36) -- (axis cs:1.25,2.55);
\addplot[thick, color=black, solid, domain=0.25:1.75, samples=2]
  {2.5 - (x-0.25)/12};
\addplot[only marks, mark=*, mark size=2pt, color=blue, forget plot]
  coordinates {(1.25,2.41667)};
\addplot[only marks, mark=*, mark size=2pt, color=red!80!black, forget plot]
  coordinates {(1.75,2.375)};
\end{axis}
\end{tikzpicture}
\end{minipage}
\caption{Mechanism in Example~\ref{ex:counter:es_stepup}, with $X_1=1/4$
fixed on $\{S\le2\}$ and $X_1(\{S=3\})=a\in[1/4,7/4]$ free.
Left: comonotonicity of $(X_1,X_2)$ with $S$ is equivalent to the
1-Lipschitz comparison
$X_1(\{S=3\})-X_1(\{S=2\})\le S(\{S=3\})-S(\{S=2\})=1$, i.e.\ to the
slope from $(2,1/4)$ to $(3,a)$ lying in $[0,1]$. The dashed line is
the slope-$1$ cone boundary; values on or below the cone
($a\in[1/4,5/4]$, gray) are comonotonic, while $a\in(5/4,7/4]$
(red) remains constraint-feasible but breaks the monotonicity of
$X_2=S-X_1$ in $S$. Right: the constrained minimum $19/8$ at $a=7/4$
falls in the non-comonotonic region; the comonotonic minimum $29/12$
at $a=5/4$ is strictly larger.}
\label{fig:counter:es_stepup}
\end{figure}

The mechanism behind Example~\ref{ex:counter:es_stepup} is that the
constraint allows $X_1$ to jump by as much as
$X_1(\{S=3\})-X_1(\{S=2\})=7/4-1/4=3/2$ when $S$ moves from $2$ to $3$,
while the aggregate loss itself increases by only
$S(3)-S(2)=1$ across the same step. In particular, the constraint
specification fails the 1-Lipschitz comparison
$X_1(\{S=s'\})-X_1(\{S=s\})\le s'-s$ for $s'>s$: the admissible gap
$3/2$ strictly exceeds the aggregate-loss increment $1$. The complementary
allocation $X_2=S-X_1$, therefore decreases from the adverse layer to
the catastrophe layer whenever agent~1 fully exploits the catastrophe
capacity, which breaks the comonotonicity of $X_2$ with $S$. The
unrestricted constrained optimizer uses the full catastrophe capacity.
At $S=3$, raising $X_1$ by one unit increases $\ES_{1/5}(X_1)$ by
$5/12$ but lowers $\ES_{1/3}(X_2)$ by $1/2$, so shifting tail mass
onto agent~1 lowers the sum $\ES_{1/5}(X_1)+\ES_{1/3}(X_2)$. The comonotonic
class cannot reach this allocation: it can raise $X_1$ on $\{S=3\}$
only to $1+X_1$ on $\{S=2\}$, that is, only to $5/4$, losing $1/2$ of
tail capacity.

This demonstrates that $\sigma(S)$-measurability and monotonicity of
the cap in $S$ are not sufficient for convex-order solidity. The
failure mechanism is any step-up schedule with a slope exceeding one
between adjacent states. 

\subsection{Value-at-Risk ceilings}\label{ss:counter:var}

The Value-at-Risk is not convex-order consistent: a degenerate
random variable can have higher $\VaR_\alpha$ than a dispersed one with
the same mean, so a $\VaR$ ceiling is not convex-order solid. In
Example~\ref{ex:counter:var} below, the initial endowments are
comonotonic and satisfy a Solvency-II-type $\VaR_{0.995}$ constraint,
yet the constrained infimum is not attained within the comonotonic
class and the comonotonic-restricted value is strictly larger than
the unrestricted constrained value.

\begin{example}\label{ex:counter:var}
Let $(A_0,A_{1a},A_{1b},A_2)$ be disjoint events with $\p(A_0)=0.9925$,
$\p(A_{1a})=\p(A_{1b})=\p(A_2)=0.0025$, and set
$A_1:=A_{1a}\cup A_{1b}$, so that $\p(A_1)=0.005$. Define the
comonotonic endowments $\zeta_1=\zeta_2=\id_{A_1}+2\,\id_{A_2}$, so
that $S=2\,\id_{A_1}+4\,\id_{A_2}$. Take $\rho_1=\ES_{0.99}$,
$\rho_2=\ES_{0.9925}$, and
$\Cset=\{X_i\ge0,\;\VaR_{0.995}(X_i)\le 1,\;i=1,2\}$. Under autarky
$\VaR_{0.995}(\zeta_i)=1$, so the endowment allocation is comonotonic
and feasible. Nonnegativity of each share $X_i$ is part of the constraint.

One may show that the minimizers of $\rho_1(X_1)+\rho_2(X_2)$ over
$\mathbb A_2(S)$, $\mathbb A_2^{\Cset}(S)$, and
$\mathbb A_2^{\Cset}(S)\cap\mathbb A_2^{+}(S)$ are, respectively,
\[
\begin{aligned}
(X_1,X_2)&=(S,\,0), & \sum_{i=1}^{2}\rho_i(X_i)&=2,\\[2pt]
(X_1,X_2)&=(\id_{A_{1a}}+2\,\id_{A_{1b}}+4\,\id_{A_2},\ \id_{A_{1a}}),
  & \dsquareC_{i=1}^{2}\rho_i(S)&=25/12,\\[2pt]
(X_1,X_2)&=(\id_{A_1}+3\,\id_{A_2},\ \id_{A_1}+\id_{A_2}),
  & \dboxplusC_{i=1}^{2}\rho_i(S)&=9/4.
\end{aligned}
\]
Since $25/12<9/4$, the constrained comonotonic value strictly exceeds
the constrained value, establishing non-comonotonicity of the
constrained infimum.
The mechanism is that a VaR ceiling permits a non-$\sigma(S)$-measurable
split of the intermediate loss layer: $X_1^\star$ takes the value~$1$
on $A_{1a}$ and~$2$ on $A_{1b}$, even though $S=2$ on both atoms.
Comonotonicity forces each $X_i$ to be constant on each level set
of~$S$, ruling out this selective assignment. A regulator who
wants a tail-risk cap on each share can use $\ES$ in place of $\VaR$.
Among coherent, law-invariant risk measures that dominate
$\VaR_\alpha$, $\ES_\alpha$ is the smallest, so an $\ES_\alpha$ cap
at the same confidence level is the closest coherent tightening of
the original $\VaR_\alpha$ cap. Coherence of $\ES$ makes its
sublevel sets convex-order solid, so Theorem~\ref{thm:cc:constrained}
applies and the comonotonic reduction is restored.
\end{example}

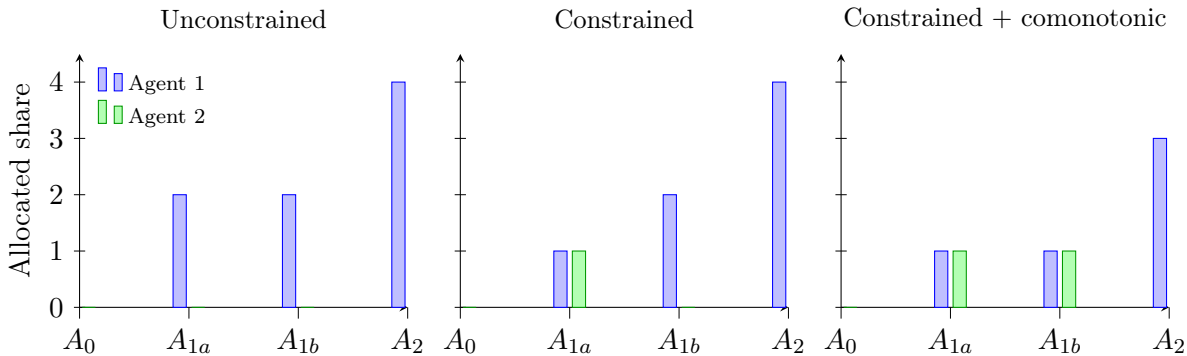
\begin{figure}[ht]
\centering
\begin{tikzpicture}
\begin{groupplot}[
  group style={
    group size=3 by 1,
    horizontal sep=0.7cm,
    ylabels at=edge left,
    yticklabels at=edge left,
  },
  width=0.36\linewidth, height=0.30\linewidth,
  symbolic x coords={A0,A1a,A1b,A2},
  xtick=data,
  xticklabels={$A_0$,$A_{1a}$,$A_{1b}$,$A_2$},
  ymin=0, ymax=4.5, ytick={0,1,2,3,4},
  ylabel={Allocated share},
  enlarge x limits={lower=0.2, upper=0.55},
  axis lines=left, tick style={black},
  title style={font=\small},
  legend style={draw=none, fill=none, font=\scriptsize},
  legend cell align=left,
]
\nextgroupplot[title={Unconstrained}, legend pos=north west, ybar, bar width=5pt]
\addplot[fill=blue!25, draw=blue] coordinates {(A0,0) (A1a,2) (A1b,2) (A2,4)};
\addlegendentry{Agent 1}
\addplot[fill=green!30, draw=green!60!black] coordinates {(A0,0) (A1a,0) (A1b,0) (A2,0)};
\addlegendentry{Agent 2}
\nextgroupplot[title={Constrained}, ybar, bar width=5pt]
\addplot[fill=blue!25, draw=blue] coordinates {(A0,0) (A1a,1) (A1b,2) (A2,4)};
\addplot[fill=green!30, draw=green!60!black] coordinates {(A0,0) (A1a,1) (A1b,0) (A2,0)};
\nextgroupplot[title={Constrained $+$ comonotonic}, ybar, bar width=5pt]
\addplot[fill=blue!25, draw=blue] coordinates {(A0,0) (A1a,1) (A1b,1) (A2,3)};
\addplot[fill=green!30, draw=green!60!black] coordinates {(A0,0) (A1a,1) (A1b,1) (A2,1)};
\end{groupplot}
\end{tikzpicture}
\caption{The three minimizers in Example~\ref{ex:counter:var}, shown on
the four atoms ($\p(A_0)=0.9925$, $\p(A_{1a})=\p(A_{1b})=\p(A_2)=0.0025$).
The constrained optimum splits the layer $A_1$ asymmetrically across
$A_{1a}$ and $A_{1b}$; the comonotonic restriction forces equal shares
on these two atoms and pays a strictly larger total risk
$9/4>25/12$.}
\label{fig:counter:var}
\end{figure}

The splitting of the level set $\{S=2\}=A_{1a}\cup A_{1b}$ requires
the probability space to be rich enough that this level set is not a
single atom; on an atomless space, any level set of positive measure
can be split, and the example applies whenever $\p(S=2)>0$.

Unconstrained problems
with convex-order-consistent preferences admit comonotonic optima by
the comonotonic improvement theorem,
and distortion-riskmetric settings can exhibit comonotonic and
counter-monotonic structures coexisting on different quantile segments
\citep{lauzier2026risk}. The feasible allocation constructed above
splits the level set $\{S=2\}$ across $A_{1a}$ and~$A_{1b}$, so it is
not $\sigma(S)$-measurable and cannot arise from any common
quantile-segment reordering of the two shares. $\VaR$ constraints
therefore admit feasible allocations that strictly improve on every
comonotonic allocation and whose dependence structure is more intricate
than a quantile-segment mixture.

One further inadmissible family involves lower bounds on
convex-order-consistent functionals. Regulatory minimum capital or
reserve charges that use a convex-order-monotone risk measure fit this
pattern, including the Expected Shortfall charge on market risk under
the Basel~III FRTB and per-member reserve floors computed via
variance-based or concave-distortion premium principles. 
Two-sided tolerance intervals around a target yield the
same form $\varrho_i(X_i)\ge c_i$ with $\varrho_i$ convex-order
monotone. Because convex-order reduction decreases $\varrho_i$, such
lower bounds are not convex-order solid. The only admissible lower
bound on a marginal-law functional is on the expectation, which convex
order preserves exactly.

%% =========================================================
\section{Constraint families admitting comonotonic optima}\label{sec:sufficient}

This section gathers the constraint families most relevant to risk
management and regulation that generate convex-order-solid feasible
sets and therefore admit the reduction of
Theorem~\ref{thm:cc:constrained}. We collect four marginal constraint
families whose induced feasible sets are componentwise convex-order
solid in Theorem~\ref{thm:solid_families}. Convex piecewise-linear
funding-cost adjustments folded into the objective are handled
separately at the end of this section.

Before turning to concrete families, we record a one-dimensional
version of solidity that we use to verify each family componentwise.

\begin{definition}[Marginally convex-order solid]\label{def:marg_solid}
A set $\mathfrak{C}\subseteq L^1$ is \emph{marginally convex-order
solid} if $Z\in\mathfrak{C}$ and $W\cx Z$ imply $W\in\mathfrak{C}$.
\end{definition}

If each component constraint $X_i\in\mathfrak{C}_i$ is marginally
convex-order solid, then the feasible set
$\{\bX\in\mathbb{A}_n(S):X_i\in\mathfrak{C}_i\ \text{for every }i\}$
is componentwise convex-order solid in the sense of
Definition~\ref{def:cc:solid}. This is the mechanism through which
Theorem~\ref{thm:solid_families} below combines its four
items; intersections close under both marginal and componentwise
solidity, as established in Section~\ref{sec:solidity}.

\begin{theorem}[Solid constraint families]\label{thm:solid_families}
Each of the following constraint sets is marginally convex-order
solid, and componentwise convex-order solidity of the induced
feasible set follows from the paragraph preceding the theorem,
with closure under intersections noted in
Section~\ref{sec:solidity}:
\begin{enumerate}
\item[(a)] Deterministic pathwise bounds
$\{Z\in L^1:\ \ell\le Z\le u\ \text{a.s.}\}$ with
$\ell\in[-\infty,+\infty)$, $u\in(-\infty,+\infty]$, $\ell\le u$.
\item[(b)] Linear expectation constraints $\{\E[Z]\le c\}$,
$\{\E[Z]=c\}$, $\{\E[Z]\ge c\}$.
\item[(c)] Orlicz-type bounds $\{\E[\varphi(Z)]\le c\}$ for convex
$\varphi:\R\to\R$ bounded below, including $L^p$-norm and
entropic bounds; variance bounds $\{\mathrm{Var}(Z)\le v\}$ follow
from $\varphi(x)=x^2$ together with preservation of the mean.
\item[(d)] Risk-measure ceilings $\{\varrho(Z)\le c\}$ for any
convex-order-consistent $\varrho:L^1\to\R\cup\{+\infty\}$, $c\in\R$.
\end{enumerate}
\end{theorem}

% \begin{proof}
% Each item uses a single convex test function in
% Definition~\ref{def:cx}: (a) the hinges $(x-u)^+$ and $(\ell-x)^+$;
% (b) the identity, preserved exactly; (c) the given $\varphi$;
% (d) convex-order consistency of $\varrho$ directly.
% \end{proof}

\begin{proof}
We verify that each family is marginally convex-order solid; the
componentwise statement for the induced feasible set was recorded in
the paragraph preceding the theorem. In each case, we fix
$Z\in\mathfrak{C}$ and $W\cx Z$, and check that $W\in\mathfrak{C}$.

For item~(a), suppose first that $u<+\infty$ and consider the convex
hinge $\phi_u(x)=(x-u)^+$. Because $Z\le u$ almost surely, one has
$\phi_u(Z)=0$ a.s., so $\E[\phi_u(Z)]=0$. Definition~\ref{def:cx}
applied to $\phi_u$ then yields $\E[\phi_u(W)]\le 0$, and since
$\phi_u\ge 0$ this forces $\phi_u(W)=0$ a.s., that is, $W\le u$ a.s.
The lower bound $W\ge\ell$ a.s.\ follows by the same argument with
$\phi_\ell(x)=(\ell-x)^+$ whenever $\ell>-\infty$. If $u=+\infty$ or
$\ell=-\infty$, the corresponding inequality is vacuous.

For item~(b), both $x\mapsto x$ and $x\mapsto -x$ are convex, so
Definition~\ref{def:cx} yields $\E[W]\le\E[Z]$ and $-\E[W]\le-\E[Z]$,
hence $\E[W]=\E[Z]$. The three sets $\{\E[Z]\le c\}$, $\{\E[Z]=c\}$
and $\{\E[Z]\ge c\}$ are therefore each preserved under $\cx$.

For item~(c), the test function $\varphi$ is convex by hypothesis,
so Definition~\ref{def:cx} gives $\E[\varphi(W)]\le\E[\varphi(Z)]\le c$;
the boundedness below of $\varphi$ ensures that the extended
expectations are well defined. The variance bound follows by taking
$\varphi(x)=x^2$: combining $\E[W^2]\le\E[Z^2]$ with $\E[W]=\E[Z]$
from item~(b) gives $\mathrm{Var}(W)\le\mathrm{Var}(Z)\le v$.

For item~(d), convex-order consistency of $\varrho$ is exactly the
implication $W\cx Z\Rightarrow\varrho(W)\le\varrho(Z)$, so
$\varrho(W)\le\varrho(Z)\le c$ directly.
\end{proof}

Item~(d) applies to every law-invariant coherent risk measure and to
every distortion risk measure with a concave distortion, including
Expected Shortfall, spectral measures, Wang transforms with concave
generator, and proportional hazard transforms; see
\citet{wang1998ordering} and \citet{dhaene2002actuarial} for the
convex-order properties of distortion risk measures.
Expected Shortfall ceilings $\ES_{\alpha_i}(X_i)\le R_i$, such as
those used in the Basel FRTB framework and the Swiss Solvency Test,
therefore generate convex-order-solid feasibility sets, while VaR-based
regimes such as the Solvency~II standard capital requirement fall
under the caveat in Section~\ref{ss:counter:var}. Heterogeneity of
$\alpha_i$ across firms does not break solidity.

\begin{corollary}[Capacity caps]\label{cor:caps}
For deterministic constants $K_i\in[0,+\infty]$, the feasibility set
$\{\bX\in\mathbb{A}_n(S):\ 0\le X_i\le K_i\ \text{a.s.\ for every } i\}$
is componentwise convex-order solid, so
Theorem~\ref{thm:cc:constrained} applies under the standing
hypotheses whenever this set is nonempty.
\end{corollary}

Corollary~\ref{cor:caps} covers limited-liability clauses,
per-member capacity limits in reinsurance pools, layer widths in
excess-of-loss reinsurance programmes, and public--private pool structures with
heterogeneous governmental backstop caps. The analysis of inadmissible constraints in
Section~\ref{sec:counter} shows why the cap $K_i$ must be a
deterministic constant rather than a function of $S$ with over-unit
slope, and rules out caps depending on state variables outside
$\sigma(S)$.

Aggregate-indexed pathwise bounds $\ell(S)\le X_i\le u(S)$ with
nondecreasing $\ell,u$ are not automatically admissible:
Example~\ref{ex:counter:es_stepup} shows that monotonicity alone does
not prevent the comonotonic rearrangement from violating the envelope.
Such caps should be treated as potentially inadmissible unless a
separate solidity argument is supplied.

Premium budgets $\E[X_i]\le P_i$ and actuarial-fairness constraints
$\E[X_i]=\E[\zeta_i]$, which equate each agent's expected
post-allocation loss with its expected initial loss, are admissible
by item~(b); equality and lower-bound forms are admissible as well, a
special feature of the linearity of $\E$.

Each agent's preferences may combine a baseline risk measure with a
convex loading of the post-allocation loss without disturbing
solidity of the feasible set. In
tiered-capital or cost-of-funds models, agent $i$'s objective takes
the form $\rho_i(X_i)+\E[\varphi_i(X_i)]$, where $\varphi_i:\R\to\R$ is
convex, nondecreasing, piecewise linear, and bounded below (so that
$\E[\varphi_i(X_i)]$ is an extended expectation with integrable
negative part, taking values in $\R\cup\{+\infty\}$). A typical
specification is
$\varphi_i(x)=\alpha_i(x-R_i)^+ +(\beta_i-\alpha_i)(x-R_i-B_i)^+$,
with $R_i\in\R$ a retention threshold, $B_i\ge0$ a layer width, and
$0\le\alpha_i\le\beta_i$ so that the slopes are nonnegative and
nondecreasing, encoding a reserve, collateral, and fire-sale ladder. The
adjusted objective is convex-order consistent: convex-order consistency
of $X\mapsto\E[\varphi_i(X)]$ follows by applying
Definition~\ref{def:cx} to the convex test function $\varphi_i$, and
summing with the convex-order-consistent $\rho_i$ preserves the
property, so Theorem~\ref{thm:cc:constrained} applies directly on any
convex-order-solid feasibility set.

%% =========================================================
\section{Application: constrained mean-variance risk sharing}\label{sec:mv}

We illustrate Theorem~\ref{thm:cc:constrained} on the mean-variance
risk-sharing problem. The unconstrained mean-variance planner admits a
well-known affine quota-share optimum. Under deterministic caps and
floors, Theorem~\ref{thm:solid_families}(a) (and
Corollary~\ref{cor:caps} in the special case $L_i=0$) licenses the
comonotonic reduction, and
the constrained optimum takes a truncated-affine form driven by a single
statewise shadow price. Hard caps produce a piecewise linear sharing
rule with slopes determined by the active set of uncapped agents.

\subsection{Unconstrained mean-variance sharing}

Fix risk-aversion parameters $\delta_1,\dots,\delta_n>0$, and assume
$S\in L^2$ throughout this section. Each agent $i$ evaluates its
share through the mean-variance risk measure
\begin{equation}\label{eq:mv:rho}
\rho_i(X)\ :=\ \E[X]+\delta_i\,\mathrm{Var}(X),\qquad X\in L^2.
\end{equation}
The functional $\rho_i$ is law-invariant and convex-order consistent:
if $Y\cx X$ with $X\in L^2$, then Definition~\ref{def:cx} applied with
$\varphi(x)=x^2$ gives $\E[Y^2]\le\E[X^2]<\infty$, so $Y\in L^2$,
and combined with $\E[Y]=\E[X]$ this yields
$\mathrm{Var}(Y)\le\mathrm{Var}(X)$ and therefore $\rho_i(Y)\le\rho_i(X)$.
Assumption~\ref{ass:cx_monotone} is in force on the domain $L^2$.

The clearing condition forces $\sum_i\E[X_i]=\E[S]$, so the
unconstrained risk-sharing problem
\begin{equation}\label{eq:mv:unconstrained}
\inf_{\bX\in\mathbb{A}_n(S)\cap(L^2)^n}\sum_{i=1}^n\rho_i(X_i)
\end{equation}
is equivalent, up to the additive constant $\E[S]$, to the pure
variance-minimization problem
\[
\inf_{\bX\in\mathbb{A}_n(S)\cap(L^2)^n}\sum_{i=1}^n\delta_i\,\mathrm{Var}(X_i).
\]
The following proposition is the mean-variance instance of the
proportional sharing rule for dilated risk measures, noted in
\citet[Remark~4.6]{acciaio2007optimal}; see
\citet[Theorem~3.9]{barrieu2005inf} for a proof, and
\citet{bernard2025risk} for a related statement. We state it to fix
the notation used in the constrained problem.

\begin{proposition}\label{prop:mv:unconstrained}
The unconstrained mean-variance problem \eqref{eq:mv:unconstrained}
admits the optimum
\begin{equation}\label{eq:mv:opt:affine}
X_i^{\star}\ =\ \frac{\delta_i^{-1}}{\sum_{j=1}^n\delta_j^{-1}}\,S\,+\,\beta_i,
\end{equation}
where $\beta_1,\dots,\beta_n\in\R$ with $\sum_i\beta_i=0$;
the optimizer is unique $\p$-a.s.\ up to cash transfers $\bm{\beta}$.
\end{proposition}

\subsection{Constrained mean-variance sharing}\label{ss:mv:constrained}

Impose deterministic capacity bounds on each share:
\begin{equation}\label{eq:mv:caps}
L_i\ \le\ X_i\ \le\ U_i\qquad\text{a.s.,}\qquad i\in[n],
\end{equation}
with $L_i\in[-\infty,+\infty)$, $U_i\in(-\infty,+\infty]$, $L_i<U_i$,
and $\sum_iL_i\le S\le\sum_iU_i$ a.s.\ in the extended sense. Let
\[
\mathbb{A}_n^{\mathrm{cap}}(S)\ :=\ \left\{\bX\in\mathbb{A}_n(S)\cap(L^2)^n:\ L_i\le X_i\le U_i\ \text{a.s.}\ \forall i\right\}.
\]
The constrained mean-variance problem is
\begin{equation}\label{eq:mv:constrained}
\inf_{\bX\in\mathbb{A}_n^{\mathrm{cap}}(S)}\sum_{i=1}^n\rho_i(X_i).
\end{equation}

Applied componentwise, Theorem~\ref{thm:solid_families}(a) shows that
$\mathbb{A}_n^{\mathrm{cap}}(S)$ is componentwise convex-order solid in
the sense of Definition~\ref{def:cc:solid}, and convex order preserves
$L^2$ membership, so the comonotonic improvement stays in $(L^2)^n$.
Theorem~\ref{thm:cc:constrained} therefore yields the reduction
\[
\inf_{\bX\in\mathbb{A}_n^{\mathrm{cap}}(S)}\sum_i\rho_i(X_i)
\ =\ \inf_{\bX\in\mathbb{A}_n^{\mathrm{cap}}(S)\cap\mathbb{A}_n^{+}(S)}\sum_i\rho_i(X_i),
\]
so we may restrict attention to comonotonic allocations of the form
$X_i=f_i(S)$ with $f_i:\R\to[L_i,U_i]$ nondecreasing and
$\sum_if_i(s)=s$ for $F_S$-a.e.\ $s$. Let
$\mathcal{S}:=\mathrm{supp}(F_S)$ denote the topological support of
$S$, that is, the smallest closed set $K\subseteq\R$ with
$\p(S\in K)=1$. Since the optimum has the form $X_i^\star=f_i^\star(S)$,
the functions $f_i^\star$ are pinned down only on $\mathcal{S}$;
statements below about the shape of $s\mapsto f_i^\star(s)$ therefore
refer to $s\in\mathcal{S}$ and hold up to $F_S$-null sets.

The following theorem identifies the form of the constrained optimum.

\begin{theorem}\label{thm:mv:form}

Any solution $\bX^{\star}$ to problem \eqref{eq:mv:constrained} has the form
\begin{equation}\label{eq:mv:form}
X_i^{\star}\ =\ \Pi_i\!\left(c_i^{\star}+\delta_i^{-1}\,\eta^{\star}(S)\right)\quad\text{a.s.},\qquad
\Pi_i(x)\ :=\ \min\!\left\{U_i,\ \max\!\left\{L_i,\ x\right\}\right\},
\end{equation}
with $c_i^{\star}=\E[X_i^{\star}]$ and a nondecreasing
$\eta^{\star}:\mathcal{S}\to\R$ defined by the clearing condition
$$\sum_i\Pi_i(c_i^{\star}+\delta_i^{-1}\eta^{\star}(s))=s.$$ 
Furthermore, the map $s\mapsto f_i^{\star}(s)$ is
affine with slope
\begin{equation}\label{eq:mv:slopes}
\frac{\d f_i^{\star}}{\d s}\ =\
\begin{cases}
\displaystyle\frac{\delta_i^{-1}}{\sum_{j\in A}\delta_j^{-1}},&i\in A,\\[0.9em]
0,&i\notin A
\end{cases}
\end{equation} 
on any subinterval of $\mathcal{S}$ along which the active set
$A(s):=\{i\in[n]:L_i<f_i^{\star}(s)<U_i\}$ is constant and equal to some nonempty $A\subseteq[n]$.
\end{theorem}

% \begin{proof}
% By Theorem~\ref{thm:cc:constrained}(i) an optimizer can be chosen
% comonotonic. Setting $c_i^\star=\E[X_i^\star]$ and using
% $\mathrm{Var}(X_i)\le\E[(X_i-c_i^\star)^2]$ with equality at the
% optimizer, the problem reduces statewise (conditional on $S=s$) to
% the strictly convex quadratic program in
% Remark~\ref{rem:mv:technical}. Dualizing the clearing constraint with
% multiplier $2\eta$ decouples the box minimization and gives the
% truncated-affine form \eqref{eq:mv:form}. The slope formula
% \eqref{eq:mv:slopes} follows because on any interval where the
% active set $A$ is constant, the clearing identity
% $\sum_if_i^\star(s)=s$ forces $\eta^\star$ affine in~$s$ with the
% stated slope. When all bounds are infinite, $\Pi_i$ is the identity
% and the formula reduces to the unconstrained expression.
% \end{proof}

\begin{proof}
By Theorem~\ref{thm:solid_families}(a), applied componentwise, the
constraint set $\mathbb{A}_n^{\mathrm{cap}}(S)$ is componentwise
convex-order solid in the sense of Definition~\ref{def:cc:solid}.
Theorem~\ref{thm:cc:constrained}(i) then produces a feasible
allocation comonotonic with $S$ that componentwise dominates any
given optimizer in convex order; since each $\rho_i$ is consistent
with convex order, this comonotonic allocation is itself optimal. Fix such an optimizer, write it
as $X_i^{\star}=f_i^{\star}(S)$, and set
$c_i^{\star}:=\E[X_i^{\star}]$. Since $\sum_i\E[X_i]=\E[S]$ is fixed
across $\mathbb{A}_n^{\mathrm{cap}}(S)$, the expectation part of the
objective in \eqref{eq:mv:constrained} is constant, and only the
variance term remains.

We first show that, for $F_S$-a.e.\ $s\in\mathcal{S}$, the vector
$f^{\star}(s)=(f_1^{\star}(s),\ldots,f_n^{\star}(s))$ solves the
finite-dimensional statewise program
\begin{equation}\label{eq:mv:inner}
\min_{x\in\R^n}\sum_{i=1}^n\delta_i(x_i-c_i^{\star})^2
\quad\text{s.t.}\quad\sum_{i=1}^n x_i=s,\ L_i\le x_i\le U_i.
\tag*{($Q_s$)}
\end{equation}
If this failed on a set of $s$-values with positive $F_S$-mass,
replacing $f^{\star}(s)$ on that set by the unique minimizer of
$(Q_s)$ would produce another feasible allocation
$Y_i=g_i(S)\in\mathbb{A}_n^{\mathrm{cap}}(S)$ with
\[
\sum_i\rho_i(Y_i)
\ =\ \E[S]+\sum_i\delta_i\,\mathrm{Var}(Y_i)
\ \le\ \E[S]+\sum_i\delta_i\,\E[(Y_i-c_i^{\star})^2],
\]
whereas, using
$\mathrm{Var}(X_i^{\star})=\E[(X_i^{\star}-c_i^{\star})^2]$ at the
optimizer,
\[
\sum_i\rho_i(X_i^{\star})
\ =\ \E[S]+\sum_i\delta_i\,\E[(X_i^{\star}-c_i^{\star})^2].
\]
A statewise improvement on a positive-mass set makes the right-hand
side strictly smaller, contradicting optimality of $\bX^{\star}$.
Hence $f^{\star}(s)$ solves $(Q_s)$ for $F_S$-a.e.\ $s$.

It remains to solve $(Q_s)$. This is a strictly convex quadratic
program with affine and box constraints, so the KKT conditions are
necessary and sufficient. With Lagrangian
\[
\mathcal{L}(x,\eta,\lambda,\mu)
\ =\ \sum_{i=1}^n\delta_i(x_i-c_i^{\star})^2
-2\eta\!\left(\sum_{i=1}^n x_i-s\right)
+\sum_{i=1}^n\lambda_i(L_i-x_i)
+\sum_{i=1}^n\mu_i(x_i-U_i),
\]
the KKT system
\[
\sum_i x_i=s,\qquad L_i\le x_i\le U_i,\qquad
2\delta_i(x_i-c_i^{\star})-2\eta-\lambda_i+\mu_i=0,
\]
\[
\lambda_i(x_i-L_i)=0,\qquad\mu_i(U_i-x_i)=0,\qquad
\lambda_i,\mu_i\ge0,
\]
is equivalent to
\[
x_i\ =\ \Pi_i\!\left(c_i^{\star}+\delta_i^{-1}\eta\right),\qquad
\Pi_i(x)=\min\{U_i,\max\{L_i,x\}\}.
\]
The clearing condition $\sum_i x_i=s$ then determines the multiplier
$\eta^{\star}(s)$ through
\[
H(\eta^{\star}(s))\ =\ s,\qquad
H(\eta)\ :=\ \sum_{i=1}^n\Pi_i\!\left(c_i^{\star}+\delta_i^{-1}\eta\right),
\]
and, since $H$ is continuous and nondecreasing, a nondecreasing
measurable selector $s\mapsto\eta^{\star}(s)$ exists. Combining,
\[
X_i^{\star}\ =\ \Pi_i\!\left(c_i^{\star}+\delta_i^{-1}\eta^{\star}(S)\right)\quad\text{a.s.,}
\]
which is the truncated-affine form \eqref{eq:mv:form}.

Finally, on any subinterval of $\mathcal{S}$ along which the active
set $A=\{i\in[n]:L_i<f_i^{\star}(s)<U_i\}$ is constant, the inactive
shares are fixed at their bounds and the clearing condition reduces
to
\[
\sum_{i\in A}\left(c_i^{\star}+\delta_i^{-1}\eta^{\star}(s)\right)
+\sum_{i\notin A}f_i^{\star}(s)\ =\ s,
\]
so $\eta^{\star}$ is affine in $s$ on that interval. Differentiating
yields
\[
\frac{\d f_i^{\star}}{\d s}
\ =\ \frac{\delta_i^{-1}}{\sum_{j\in A}\delta_j^{-1}}\ \text{for }i\in A,
\qquad
\frac{\d f_i^{\star}}{\d s}\ =\ 0\ \text{for }i\notin A,
\]
which is \eqref{eq:mv:slopes}. When $L_i=-\infty$ and $U_i=+\infty$
for every $i$, $\Pi_i$ is the identity and $A=[n]$, and the formula
reduces to the unconstrained affine allocation
\eqref{eq:mv:opt:affine}.
\end{proof}
\begin{remark}[Existence and measurability of the shadow-price selector]\label{rem:mv:technical}
Theorem~\ref{thm:mv:form} asserts that every comonotonic optimum has the
truncated-affine form \eqref{eq:mv:form} through a nondecreasing
measurable function $\eta^{\star}:\mathcal{S}\to\R$ of the
aggregate loss. Two facts must hold for this claim to make sense:
the statewise minimizer must be unique, and the multiplier that
generates it must be choosable as a Borel function of $s$, even though
the clearing condition does not pin $\eta$ down uniquely when several
$\eta$-values give the same total. We record both here.

For fixed $s\in\mathcal{S}$, the statewise program
\ref{eq:mv:inner} of the proof of Theorem~\ref{thm:mv:form} has a
strictly convex, coercive objective and a convex polyhedral
feasible set, so it admits a unique minimizer $\hat x(s)$. Dualizing
the adding-up constraint yields
\begin{equation}\label{eq:mv:clip}
\hat x_i(s)\ =\ \Pi_i\!\left(c_i^{\star}+\delta_i^{-1}\,\eta(s)\right)
\end{equation}
for any $\eta(s)$ satisfying $H(\eta(s))=s$, where
$H(\eta):=\sum_i\Pi_i(c_i^{\star}+\delta_i^{-1}\eta)$ is continuous and
nondecreasing. The preimage $H^{-1}(\{s\})$ is therefore a nonempty
closed interval, possibly unbounded at boundary values of $s$, and
$\hat x(s)$ is constant along it, so the non-uniqueness of $\eta(s)$
does not affect the allocation on any flat region of $H$. Any
nondecreasing measurable selector $s\mapsto\eta^{\star}(s)$ of the
multivalued inverse of $H$ therefore defines a finite, single-valued
and measurable $\eta^{\star}:\mathcal{S}\to\R$, so
$\hat X_i:=\hat x_i(S)$ is a Borel allocation comonotonic with $S$ of
the form claimed by Theorem~\ref{thm:mv:form}.
\end{remark}

Theorem~\ref{thm:mv:form} records three structural features. Globally,
the allocation is truncated-affine through a single nondecreasing
statewise shadow price $\eta^{\star}(S)$, pinned down by the adding-up
constraint. Locally, sharing is proportional to risk tolerances among
active agents, reproducing the unconstrained slopes
$\delta_i^{-1}/\sum_{j\in A}\delta_j^{-1}$ on each region where the active set $A$ is
constant. Capped agents contribute zero marginal sharing on that
region, and the remaining active agents re-share incremental losses
with renormalized slopes; slope discontinuities in
$s\mapsto f_i^{\star}(s)$ occur at levels where an agent transitions
between interior and boundary.

\subsection{Two-agent example}

Consider $n=2$ agents with risk-aversion parameters $\delta_1,\delta_2>0$.
Agent 1 faces a hard upper cap $U_1=C>0$ and a lower bound
$L_1=0$; agent 2 is uncapped in the sense that $L_2=-\infty$ and
$U_2=+\infty$. Assume $S\ge 0$ a.s.
The feasible set is
\[
\mathbb{A}_2^{\mathrm{cap}}(S)\ =\ \left\{(X_1,X_2):\ X_1+X_2=S,\ 0\le X_1\le C\right\}.
\]
Write $a:=\delta_1^{-1}/(\delta_1^{-1}+\delta_2^{-1})\in(0,1)$ for the unconstrained quota
share of agent 1, and assume the cap binds on a nontrivial set of
states, that is, $\p(aS>C)>0$, equivalently $\esssup S>C/a$.

\begin{corollary}\label{cor:mv:two}
The constrained mean-variance optima in the two-agent pool above take
the form
\begin{equation}\label{eq:mv:two:agent}
X_1^{\star}\ =\ \min\!\left\{C,\ \max\!\left\{0,\ aS+\beta^{\star}\right\}\right\},\qquad
X_2^{\star}\ =\ S-X_1^{\star}\quad\text{a.s.,}
\end{equation}
where $\beta^{\star}\in\R$ is any solution of the fixed-point condition
\begin{equation}\label{eq:mv:two:fp}
\beta^{\star}\ =\ \E\!\left[\min\!\left\{C,\ \max\!\left\{0,\ aS+\beta^{\star}\right\}\right\}\right]-a\,\E[S].
\end{equation}
The fixed-point equation \eqref{eq:mv:two:fp} admits at least one
solution $\beta^{\star}\in\R$, and any such $\beta^{\star}$ yields
an optimum of the form \eqref{eq:mv:two:agent}; the solution set
is a closed interval $[\beta_-,\beta_+]$ of deterministic balanced
cash transfers.
\end{corollary}

\begin{proof}
Agent~2 is uncapped, so $\Pi_2$ is the identity in
Theorem~\ref{thm:mv:form}. On the states where the active set
equals $\{1,2\}$, the slope formula \eqref{eq:mv:slopes} gives
$\d X_1^{\star}/\d s=a$; on the cap regions $X_1^{\star}$ is pinned
at $0$ or $C$. Continuity of $s\mapsto f_1^{\star}(s)$ forces
\[
X_1^{\star}\ =\ \Pi(aS+\beta^{\star}),\qquad
\Pi(y):=\min\{C,\max\{0,y\}\},
\]
for some $\beta^{\star}\in\R$ identified from the interior identity
$X_1^{\star}=aS+\beta^{\star}$ as $\beta^{\star}=c_1^{\star}-a\E[S]$,
using $c_1^{\star}+c_2^{\star}=\E[S]$. Combining with
$c_1^{\star}=\E[X_1^{\star}]=\E[\Pi(aS+\beta^{\star})]$ yields the
fixed-point equation \eqref{eq:mv:two:fp}. Existence of
$\beta^{\star}$ follows from existence of the optimum.

Let $h(\beta):=\E[\Pi(aS+\beta)]$. Since $\Pi$ is nondecreasing and
$1$-Lipschitz, so is $h$, and $h(\beta)-\beta-a\E[S]$ is continuous
and nonincreasing, so the solution set of \eqref{eq:mv:two:fp} is
a closed interval $[\beta_-,\beta_+]$.
\end{proof}

The allocation \eqref{eq:mv:two:agent} is the classical
quota-share-with-limit rule up to a deterministic transfer. 
Below the attachment $s=(C-\beta^{\star})/a$
the two agents pool loss in proportion to their tolerances, and above
the attachment the capped agent stops absorbing loss while the
uncapped agent takes the full excess. When the lower boundary
$\{X_1^{\star}=0\}$ does not bind (for instance whenever
$\beta^{\star}\ge 0$ and $S\ge 0$ a.s.), the single breakpoint at
$s=(C-\beta^{\star})/a$ is the only slope change and corresponds to
agent~$1$ transitioning from interior to upper cap; below it the
slopes are $(a,1-a)$, above it they are $(0,1)$. The formula
reproduces the limited-quota-share treaty documented in the standard
reinsurance reference \citet{albrecher2017reinsurance}, up to the
transfer $\beta^{\star}$ that the paper's convention usually absorbs
into a premium. The comonotonic reduction in
Theorem~\ref{thm:cc:constrained}, instantiated by
Theorem~\ref{thm:solid_families}(a) componentwise
(Corollary~\ref{cor:caps} covers the symmetric $L_i=0$ case), is
what licenses restricting to this shape in the first place.

\subsection{Four-agent illustration}\label{ss:mv:numerical}

We display a four-agent allocation with deterministic
capacity caps to illustrate the successive-saturation pattern that
Theorem~\ref{thm:mv:form} produces. The primitives are
mean-variance preferences
$\rho_i(X)=\E[X]+\delta_i\mathrm{Var}(X)$ with
$\boldsymbol{\delta}=(2,3,5,6)$, unconstrained quota-share slopes
$\boldsymbol{a}=(5/12,5/18,1/6,5/36)$ from
Proposition~\ref{prop:mv:unconstrained}, and upper caps
$\boldsymbol{U}=(5,8,3,+\infty)$ with $L_i=-\infty$. Figure~\ref{fig:mv:numerical}
shows the zero-intercept representative of the unconstrained family
clipped by the caps, that is, $\widetilde X_i(s)$ starts from $a_is$
and saturates at $U_i$ as $s$ grows, with successive breakpoints
at $s_1=12$ (agent 1 saturates), $s_2=31/2$ (agent 3 saturates), and
$s_3=20$ (agent 2 saturates); between breakpoints the remaining
uncapped agents re-share with renormalized slopes. This is not the
literal optimizer of Theorem~\ref{thm:mv:form}, whose intercepts
$c_i^{\star}=\E[X_i^{\star}]$ are determined by fixed-point
conditions analogous to \eqref{eq:mv:two:fp}; the same
saturation pattern persists under the endogenous intercepts, with
breakpoints shifted accordingly.

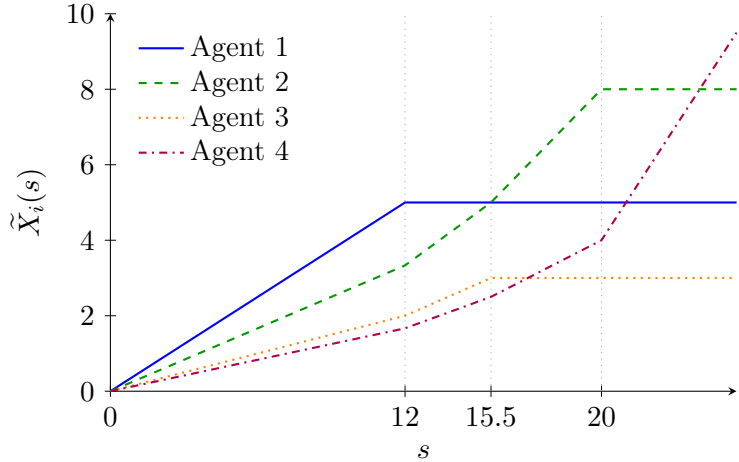
\begin{figure}[ht]
\centering
\begin{tikzpicture}
\begin{axis}[
  width=0.6\linewidth, height=0.4\linewidth,
  xlabel={$s$},
  ylabel={$\widetilde{X}_i(s)$},
  xmin=0, xmax=25.5,
  ymin=0, ymax=10,
  xtick={0,12,15.5,20},
  ytick={0,2,4,6,8,10},
  legend pos=north west,
  legend cell align=left,
  legend style={draw=none, fill=none},
  axis lines=left,
  grid=none,
  tick style={black},
  enlargelimits=false,
  clip=true,
]
% Breakpoint markers
\draw[dotted, gray!70, thin] (axis cs:12,0)   -- (axis cs:12,10);
\draw[dotted, gray!70, thin] (axis cs:15.5,0) -- (axis cs:15.5,10);
\draw[dotted, gray!70, thin] (axis cs:20,0)   -- (axis cs:20,10);

% Agent 1: (0,0) -> (12,5) -> (25.5,5)
\addplot[thick, color=blue, solid]
  coordinates {(0,0) (12,5) (25.5,5)};
\addlegendentry{Agent 1}

% Agent 2: (0,0) -> (12, 10/3) -> (15.5,5) -> (20,8) -> (25.5,8)
\addplot[thick, color=green!60!black, dashed]
  coordinates {(0,0) (12,3.3333) (15.5,5) (20,8) (25.5,8)};
\addlegendentry{Agent 2}

% Agent 3: (0,0) -> (12,2) -> (15.5,3) -> (25.5,3)
\addplot[thick, color=orange, dotted]
  coordinates {(0,0) (12,2) (15.5,3) (25.5,3)};
\addlegendentry{Agent 3}

% Agent 4: (0,0) -> (12, 5/3) -> (15.5, 5/2) -> (20,4) -> (25.5, 9.5)
\addplot[thick, color=red!70!blue, dashdotted]
  coordinates {(0,0) (12,1.6667) (15.5,2.5) (20,4) (25.5,9.5)};
\addlegendentry{Agent 4}

\end{axis}
\end{tikzpicture}
\caption{Capped quota-share allocations $\widetilde{X}_i(s)$ with
upper caps $\boldsymbol{U}=(5,8,3,+\infty)$ and zero intercepts.
Dotted vertical lines mark the breakpoints at which agents 1, 3, and
2 successively saturate their caps.}
\label{fig:mv:numerical}
\end{figure}

\subsection{Two-agent Value-at-Risk ceilings}\label{ss:mv:var}

Mean-variance risk sharing is a classical setting in which every
agent is variance-averse and the unconstrained optimum is an affine
comonotonic quota share. Even here, a $\VaR$ ceiling is enough to
break comonotonicity: the failure of convex-order solidity for $\VaR$
(Section~\ref{ss:counter:var}) translates directly into a
non-comonotonic optimum under continuous losses. To illustrate, take
two variance-averse agents with loss evaluations
$\rho_i(X)=\E[X]+\delta_i\mathrm{Var}(X)$ and parameters
$(\delta_1,\delta_2)=(0.01,1)$, so that agent~1 is much less
variance-averse than agent~2. Let $\zeta_1,\zeta_2\stackrel{\mathrm{iid}}{\sim}\mathrm{Exp}(1)$, so $S=\zeta_1+\zeta_2$ is gamma-distributed with shape~$2$ and rate~$1$,
and impose the constraint
\[
X_1+X_2=S,\qquad X_i\ge 0,\qquad \VaR_{0.95}(X_i)\le 3,\quad i=1,2,
\]
under which autarky is feasible.

Set $\lambda:=\delta_1/(\delta_1+\delta_2)=1/101$, the mean-constraint
center $m^\star\approx 0.6337$,
$a:=m^\star/(1-\lambda)\approx 0.6400$,
$r:=(3+m^\star)/(1-\lambda)\approx 3.6700$, and
$q:=\VaR_{0.95}(S)\approx 4.7439$. The constrained optimum is
$X_2^\star=f(S)$, $X_1^\star=S-f(S)$, with
\[
f(s)\ =\
\begin{cases}
s, & 0\le s\le a,\\
m^\star+\lambda s, & a<s\le r,\\
s-3, & r<s\le q,\\
m^\star+\lambda s, & s>q.
\end{cases}
\]
The map $f$ is discontinuous at $q$: $f(q-)=q-3\approx 1.74$ while
$f(q+)\approx 0.68$, so agent~2's allocation drops as $S$ crosses $q$.
The constrained objective value is
$\rho_1(X_1^\star)+\rho_2(X_2^\star)\approx 2.0517$, whereas restricting
to comonotonic allocations gives $\approx 2.0972$, confirming that the
constrained infimum is not attained within the comonotonic class. The
unconstrained optimum attains $\approx 2.0198$ and autarky (feasible,
since $\VaR_{0.95}(\zeta_i)=-\log(0.05)<3$) attains $3.01$, so the
four objective values rank as $2.0198<2.0517<2.0972<3.01$.

\begin{figure}[ht]
\centering
\begin{tikzpicture}
\begin{axis}[
  width=0.6\linewidth, height=0.4\linewidth,
  xlabel={$s$}, ylabel={$X_i^\star(s)$},
  xmin=0, xmax=7.43,
  ymin=-0.3, ymax=6.5,
  xtick={0,0.64,3.67,4.74},
  xticklabels={$0$,$a$,$r$,$q$},
  ytick={0,1,2,3,4,5,6},
  legend pos=north west,
  legend cell align=left,
  legend style={draw=none, fill=none},
  axis lines=left, grid=none, tick style={black}, enlargelimits=false, clip=true,
]
\draw[dotted, gray!70, thin] (axis cs:0.64,-0.3) -- (axis cs:0.64,6.5);
\draw[dotted, gray!70, thin] (axis cs:3.67,-0.3) -- (axis cs:3.67,6.5);
\draw[dotted, gray!70, thin] (axis cs:4.7439,-0.3) -- (axis cs:4.7439,6.5);
\draw[dotted, gray!70, thin] (axis cs:0,3) -- (axis cs:7.43,3);
% X_1^*: left branch 0 on [0,a]; affine a->r reaching 3; flat at 3 on [r,q]
\addplot[thick, color=blue, solid]
  coordinates {(0,0) (0.64002,0) (3.67002,3) (4.7439,3)};
\addlegendentry{Agent 1}
\addplot[thick, color=blue, solid, forget plot]
  coordinates {(4.7439,4.0633) (7.43,6.7227)};
% X_2^*: left branch s on [0,a]; nearly flat a->r; s-3 on [r,q]
\addplot[thick, color=green!60!black, dashed]
  coordinates {(0,0) (0.64002,0.64002) (3.67002,0.67002) (4.7439,1.7439)};
\addlegendentry{Agent 2}
\addplot[thick, color=green!60!black, dashed, forget plot]
  coordinates {(4.7439,0.6807) (7.43,0.7073)};
% Jump markers at q: filled = value attained at q (left branch), hollow = one-sided limit from the right
\addplot[only marks, mark=*, mark size=1.6, color=blue, forget plot]
  coordinates {(4.7439,3)};
\addplot[only marks, mark=o, mark size=1.6, color=blue, forget plot]
  coordinates {(4.7439,4.0633)};
\addplot[only marks, mark=*, mark size=1.6, color=green!60!black, forget plot]
  coordinates {(4.7439,1.7439)};
\addplot[only marks, mark=o, mark size=1.6, color=green!60!black, forget plot]
  coordinates {(4.7439,0.6807)};
\end{axis}
\end{tikzpicture}
\caption{Constrained mean-variance allocation under the $\VaR_{0.95}$
ceiling. At $q$, agent~2's allocation drops and agent~1's jumps up by
the same amount, so neither is comonotonic with $S$ globally.}
\label{fig:mv:var}
\end{figure}
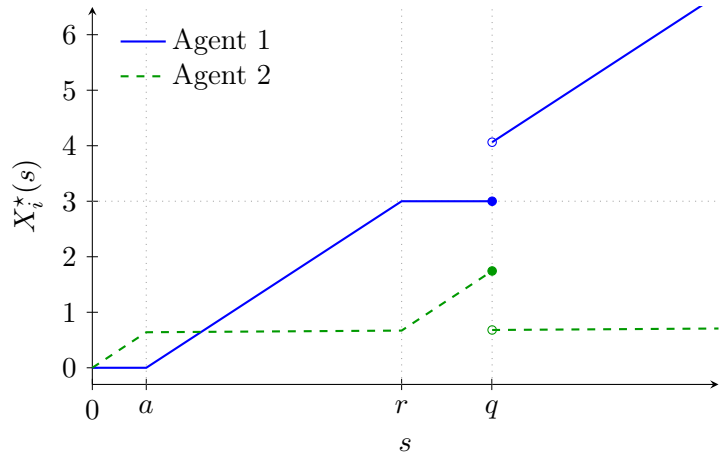

The four regimes describe how the bounds activate in turn. On $[0,a]$
the nonnegativity floor $X_1\ge 0$ binds and agent~2 absorbs the
entire loss. On $(a,r]$ both agents are interior and share loss
affinely at the unconstrained slopes $(1-\lambda,\lambda)$, with the
less-variance-averse agent~1 taking the lion's share of marginal loss.
On $(r,q]$ agent~1 is pinned at the Value-at-Risk ceiling $X_1=3$, so
agent~2 absorbs every additional unit of loss. At $s=q$ the ceiling
stops binding: since $\VaR_{0.95}$ is insensitive to the upper $5\%$
tail, the planner may transfer the accumulated cap excess back to the
less-averse agent~1 on $\{S>q\}$, and does so to restore the affine
proportional interior slopes.

The constrained allocation behaves differently across the state space:
within each of the four regimes, both $X_1^\star$ and $X_2^\star$ are
nondecreasing in $S$, so the pair is comonotonic locally, but across
the jump at $q$ the two allocations move in opposite directions, with
$X_1^\star$ rising while $X_2^\star$ falls, so pairs of states
straddling $q$ exhibit counter-monotonic behaviour. Note that here
$X_i^\star=f_i(S)$ is $\sigma(S)$-measurable, in contrast with
Example~\ref{ex:counter:var} where the optimum splits the level set
$\{S=2\}$ across $A_{1a}$ and $A_{1b}$ and so fails to be
$\sigma(S)$-measurable; the atomless exponential marginals here rule
out that mechanism, and non-comonotonicity instead shows up as the
coexistence of comonotonic pieces and counter-monotonic cross-jump
pairs. Because $\VaR_{0.95}$ does not see the upper $5\%$ tail, the
planner can still shift loss back to the less-averse agent beyond
$q$, producing a feasible allocation that strictly improves on every
comonotonic one.

%% =========================================================
\section{Conclusion}\label{sec:conclusion}

The classical unconstrained comonotonic reduction of
\citet{borch1962equilibrium}, \citet{landsberger1994comonotonic}, and
\citet{ludkovski2008comonotonicity} makes optimal risk sharing under
convex-order-consistent preferences tractable when the feasible set
is the full clearing simplex. We extend the reduction to constrained
problems. Theorem~\ref{thm:cc:constrained} shows that the comonotonic
improvement theorem survives whenever the feasible set is
componentwise convex-order solid, and
Theorem~\ref{thm:solid_families} identifies the operational
constraint families that qualify, together with their stability under
finite intersection.

Not every constraint is admissible. Sections~\ref{ss:counter:var}
and~\ref{ss:mv:var} show that $\VaR$ ceilings, which Solvency~II
imposes on insurers, escape the reduction even when every agent is
risk-averse. Risk aversion alone does not guarantee comonotonic
optima: under $\VaR$ constraints, the optimal sharing rule can combine
comonotonic and counter-monotonic regions over the state space, so the
dependence structure of the allocation departs from what the unconstrained
intuition suggests.

\section*{Acknowledgements}

CBW acknowledges financial support from the Natural Sciences and 
Engineering Research Council of Canada (RGPIN-2025-06879).

\bibliographystyle{apalike}
\bibliography{references}

@book{denneberg1994nonadditive,
  title     = {Non-{A}dditive {M}easure and {I}ntegral},
  author    = {Denneberg, D.},
  year      = {1994},
  month     = may,
  publisher = {Springer Science \& Business Media},
  isbn      = {978-0-7923-2840-7}
}

@article{dhaene2002actuarial,
  author    = {Dhaene, Jan and Denuit, Michel and Goovaerts, Marc J. and Kaas, Rob and Vyncke, David},
  title     = {The concept of comonotonicity in actuarial science and finance: theory},
  journal   = {Insurance: Mathematics and Economics},
  volume    = {31},
  number    = {1},
  pages     = {3--33},
  year      = {2002},
  publisher = {Elsevier}
}

@article{bauerle2006stochastic,
  author  = {B{\"a}uerle, Nicole and M{\"u}ller, Alfred},
  title   = {Stochastic orders and risk measures: consistency and bounds},
  journal = {Insurance: Mathematics and Economics},
  volume  = {38},
  number  = {1},
  pages   = {132--148},
  year    = {2006}
}

@article{borch1962equilibrium,
  author  = {Borch, Karl},
  title   = {Equilibrium in a reinsurance market},
  journal = {Econometrica},
  volume  = {30},
  number  = {3},
  pages   = {424--444},
  year    = {1962}
}

@article{landsberger1994comonotonic,
  author  = {Landsberger, Michael and Meilijson, Isaac},
  title   = {Co-monotone allocations, {B}ickel--{L}ehmann dispersion and the {A}rrow--{P}ratt measure of risk aversion},
  journal = {Annals of Operations Research},
  volume  = {52},
  number  = {2},
  pages   = {97--106},
  year    = {1994}
}

@article{ludkovski2008comonotonicity,
  author  = {Ludkovski, Michael and R{\"u}schendorf, Ludger},
  title   = {On comonotonicity of {P}areto optimal risk sharing},
  journal = {Statistics and Probability Letters},
  volume  = {78},
  number  = {10},
  pages   = {1181--1188},
  year    = {2008}
}

@article{wang1997axiomatic,
  author  = {Wang, Shaun S. and Young, Virginia R. and Panjer, Harry H.},
  title   = {Axiomatic characterization of insurance prices},
  journal = {Insurance: Mathematics and Economics},
  volume  = {21},
  number  = {2},
  pages   = {173--183},
  year    = {1997},
  doi     = {10.1016/S0167-6687(97)00031-0}
}

@article{wang1998ordering,
  author  = {Wang, Shaun S. and Young, Virginia R.},
  title   = {Ordering risks: {E}xpected utility theory versus {Y}aari's dual theory of risk},
  journal = {Insurance: Mathematics and Economics},
  volume  = {22},
  number  = {2},
  pages   = {145--161},
  year    = {1998},
  doi     = {10.1016/S0167-6687(97)00036-X}
}

@article{jouini2008optimal,
  author  = {Jouini, Elyes and Schachermayer, Walter and Touzi, Nizar},
  title   = {Optimal risk sharing for law invariant monetary utility functions},
  journal = {Mathematical Finance},
  volume  = {18},
  number  = {2},
  pages   = {269--292},
  year    = {2008}
}

@article{embrechts2018quantilebased,
  author  = {Embrechts, Paul and Liu, Haiyan and Wang, Ruodu},
  title   = {Quantile-based risk sharing},
  journal = {Operations Research},
  volume  = {66},
  number  = {4},
  pages   = {936--949},
  year    = {2018},
  doi     = {10.1287/opre.2017.1716}
}

@book{ruschendorf2013mathematical,
  author    = {R{\"u}schendorf, Ludger},
  title     = {Mathematical Risk Analysis: Dependence, Risk Bounds, Optimal Allocations and Portfolios},
  publisher = {Springer},
  address   = {Berlin, Heidelberg},
  year      = {2013}
}

@article{barrieu2005inf,
  author  = {Barrieu, Pauline and El Karoui, Nicole},
  title   = {Inf-convolution of risk measures and optimal risk transfer},
  journal = {Finance and Stochastics},
  volume  = {9},
  number  = {2},
  pages   = {269--298},
  year    = {2005}
}

@article{filipovic2008optimal,
  author  = {Filipovi{\'c}, Damir and Svindland, Gregor},
  title   = {Optimal capital and risk allocations for law- and cash-invariant convex functions},
  journal = {Finance and Stochastics},
  volume  = {12},
  number  = {3},
  pages   = {423--439},
  year    = {2008}
}

@article{weber2018solvency,
  author  = {Weber, Stefan},
  title   = {Solvency {II}, or how to sweep the downside risk under the carpet},
  journal = {Insurance: Mathematics and Economics},
  volume  = {82},
  pages   = {191--200},
  year    = {2018}
}

@article{xia2023infconvolution,
  author  = {Xia, Zichao and Zou, Zhenfeng and Hu, Taizhong},
  title   = {Inf-convolution and optimal allocations for mixed-{V}a{R}s},
  journal = {Insurance: Mathematics and Economics},
  volume  = {108},
  pages   = {156--164},
  year    = {2023}
}

@article{lauzier2026risk,
  author  = {Lauzier, Jean-Gabriel and Lin, Liyuan and Wang, Ruodu},
  title   = {Risk sharing, measuring variability, and distortion riskmetrics},
  journal = {Mathematical Finance},
  volume  = {36},
  number  = {2},
  pages   = {330--351},
  year    = {2026},
  doi     = {10.1111/mafi.70007}
}

@article{lauzier2023pairwise,
  author  = {Lauzier, Jean-Gabriel and Lin, Liyuan and Wang, Ruodu},
  title   = {Pairwise counter-monotonicity},
  journal = {Insurance: Mathematics and Economics},
  volume  = {111},
  pages   = {279--287},
  year    = {2023},
  doi     = {10.1016/j.insmatheco.2023.05.006}
}

@article{acciaio2007optimal,
  author  = {Acciaio, Beatrice},
  title   = {Optimal risk sharing with non-monotone monetary functionals},
  journal = {Finance and Stochastics},
  volume  = {11},
  number  = {2},
  pages   = {267--289},
  year    = {2007},
  doi     = {10.1007/s00780-007-0036-6}
}

@article{cummins2008cat,
  author  = {Cummins, J. David},
  title   = {{CAT} bonds and other risk-linked securities: state of the market and recent developments},
  journal = {Risk Management and Insurance Review},
  volume  = {11},
  number  = {1},
  pages   = {23--47},
  year    = {2008}
}

@article{bernard2025risk,
  author  = {Bernard, Carole and Vanduffel, Steven},
  title   = {Risk sharing under ambiguity},
  journal = {Working paper},
  year    = {2025}
}

@book{albrecher2017reinsurance,
  author    = {Albrecher, Hansj{\"o}rg and Beirlant, Jan and Teugels, Jozef L.},
  title     = {Reinsurance: Actuarial and Statistical Aspects},
  publisher = {Wiley},
  series    = {Wiley Series in Probability and Statistics},
  year      = {2017}
}

@techreport{bcbs2019market,
  author      = {{BCBS}},
  title       = {Minimum capital requirements for market risk},
  institution = {Basel Committee on Banking Supervision, Bank for International Settlements},
  number      = {d457},
  year        = {2019}
}

@techreport{ec2009solvency,
  author      = {{European Commission}},
  title       = {Directive 2009/138/{EC} of the {European} {Parliament} and of the {Council} of 25 {November} 2009 on the taking-up and pursuit of the business of insurance and reinsurance ({Solvency II})},
  institution = {Official Journal of the European Union},
  year        = {2009}
}

@article{denuit2020investing,
  author  = {Denuit, Michel},
  title   = {Investing in your own and peers' risks: the simple analytics of {P2P} insurance},
  journal = {European Actuarial Journal},
  volume  = {10},
  number  = {2},
  pages   = {335--359},
  year    = {2020}
}

@article{boonen2020bilateral,
  author  = {Boonen, Tim J. and Ghossoub, Mario},
  title   = {Bilateral risk sharing with heterogeneous beliefs and exposure constraints},
  journal = {ASTIN Bulletin},
  volume  = {50},
  number  = {1},
  pages   = {293--323},
  year    = {2020}
}

@article{rothschild1970increasing,
  title = {Increasing Risk: {{I}}. {{A}} Definition},
  shorttitle = {Increasing Risk},
  author = {Rothschild, Michael and Stiglitz, Joseph E},
  year = 1970,
  month = sep,
  journal = {Journal of Economic Theory},
  volume = {2},
  number = {3},
  pages = {225--243},
  issn = {0022-0531},
  doi = {10.1016/0022-0531(70)90038-4}
}

@article{carlier2003pareto,
  author  = {Carlier, Guillaume and Dana, Rose-Anne},
  title   = {{P}areto efficient insurance contracts when the insurer's cost function is discontinuous},
  journal = {Economic Theory},
  volume  = {21},
  number  = {4},
  pages   = {871--893},
  year    = {2003}
}

@article{bernard2009optimal,
  author  = {Bernard, Carole and Tian, Weidong},
  title   = {Optimal Reinsurance Arrangements under Tail Risk Measures},
  journal = {Journal of Risk and Insurance},
  year    = {2009},
  volume  = {76},
  number  = {3},
  pages   = {709--725}
}

@book{mas1995microeconomic,
  title={Microeconomic Theory},
  author={Mas-Colell, Andreu and Whinston, Michael D. and Green, Jerry R.},
  year={1995},
  publisher={Oxford University Press},
  address={New York}
}

\end{document}